\documentclass[a4paper,11pt,reqno]{amsart}

\usepackage{hyperref}       
\usepackage{url}            
\usepackage{amsthm}
\usepackage{amsfonts}
\usepackage{amsmath}
\usepackage{amssymb}
\usepackage{color}
\usepackage{multirow}
\usepackage{float}
\usepackage[noadjust]{cite}
\usepackage{mathtools}
\usepackage{blkarray}
\usepackage{booktabs}
\usepackage{hhline}
\usepackage{enumitem}
\usepackage{float}
\usepackage{longtable}
\usepackage[tableposition=below]{caption}
\usepackage[margin=2.7cm]{geometry}

\theoremstyle{definition}

\newtheorem{theorem}{Theorem}[section]
\newtheorem{corollary}[theorem]{Corollary}
\newtheorem{lemma}[theorem]{Lemma}
\newtheorem{definition}[theorem]{Definition}
\newtheorem{proposition}[theorem]{Proposition}
\newtheorem{remark}[theorem]{Remark}
\newtheorem{notation}[theorem]{Notation}

\newtheorem{example}[theorem]{Example}

\usepackage{array}
\newcolumntype{P}[1]{>{\centering\arraybackslash}p{#1}}
\usepackage{multirow, boldline,makecell}
\usepackage{enumitem}
\usepackage{MnSymbol}
\usepackage[symbol]{footmisc}

\def\C{\mathcal{C}}
\def\D{\mathcal{D}}
\def\A{\mathcal{A}}

\def\G{\mathcal{G}}
\def\E{\mathcal{E}}
\def\H{\textup{H}}
\def\mS{\mathcal{S}}

\def\ss{\textup{ss}}
\def\GL{\textup{GL}}
\def\supp{\textup{supp}}
\def\Fqnm{\mathbb{F}_q^{n\times m}}
\def\Fqknm{\mathbb{F}_q^{k\times n\times m}}

\def\Fqm{\mathbb{F}_{q^m}}

\def\Fq{\mathbb{F}_q}
\def\Tr{\textup{Tr}}
\def\dimq{\dim_{\Fq}}
\def\dimqm{\dim_{\Fqm}}

\newcommand{\Z}{\mathbb{Z}}

\newcommand{\K}{\mathbb{K}}
\newcommand{\Knn}{\mathbb{K}^{n\times n}}
\newcommand{\Knm}{\mathbb{K}^{n\times m}}
\newcommand{\Kmm}{\mathbb{K}^{m\times m}}
\newcommand{\Kknm}{\K^{k\times n\times m}}
\newcommand{\Kkmm}{\K^{k\times m\times m}}
\def\L{\mathcal{L}_{\Fqm/\Fq}}
\def\F{\mathbb{F}}

\def\<{\left<}
\def\>{\right>}
\def\diag{\textup{diag}}
\def\rk{\textup{rk}}
\def\trk{\textup{trk}}
\def\maxtrk{\textup{maxtrk}}
\def\ch{\textup{char}}

\newcommand\bigzero{\makebox(0,0){\text{\huge0}}} 
\newcommand\Tstrut[1]{\rule{0pt}{#1}}  
\newcommand\Bstrut[1]{\rule[#1]{0pt}{0pt}}

\newcommand*{\temp}[1]{\multicolumn{1}{c|}{#1}}
\newcommand*{\tempp}[1]{\multicolumn{1}{|c|}{#1}}
\newcommand*{\temppp}[1]{\multicolumn{1}{|c}{#1}}

\usepackage[dvipsnames,x11names,table]{xcolor}
\newcommand\cblue{\cellcolor{DodgerBlue1!20}} 
\newcommand\cor{\cellcolor{Orange1!20}}
\newcommand\cgr{\cellcolor{OliveDrab2!40}}

\newcolumntype{C}[1]{>{\centering\arraybackslash$}p{#1}<{$}}
\newcolumntype{L}[1]{>{\arraybackslash$}p{#1}<{$}}
\newcommand\scalemath[2]{\scalebox{#1}{\mbox{\ensuremath{\displaystyle #2}}}}

\title{\textbf{Bilinear Complexity of 3-Tensors Linked to Coding Theory}}
\author[Eimear Byrne]{Eimear Byrne}
\address{School of Mathematics and Statistics, University College Dublin, Belfield, Ireland}
\curraddr{}
\email{ebyrne@ucd.ie}
\thanks{}

\author[Giuseppe Cotardo]{Giuseppe Cotardo$^*$}
\address{School of Mathematics and Statistics, University College Dublin, Belfield, Ireland}
\curraddr{}
\email{giuseppe.cotardo@ucdconnect.ie}
\thanks{$^*$The author was supported by the Irish Research Council, grant n. GOIPG/2018/2534.}

\begin{document}

\begin{abstract}
   A well studied problem in algebraic complexity theory is the determination of the complexity of problems relying on evaluations of bilinear maps. One measure of the complexity of a bilinear map (or 3-tensor) is the optimal number of non-scalar multiplications required to evaluate it. This quantity is also described as its tensor rank, which is the smallest number of rank one matrices
   whose span contains its first slice space.
   In this paper we derive upper bounds on the tensor ranks of certain classes of $3$-tensors and give explicit constructions of sets of rank one matrices containing their first slice spaces.
   We also show how these results can be applied in coding theory to derive upper bounds on the tensor rank of some rank-metric codes. In particular, we compute the tensor rank of some families of $\Fqm$-linear codes and we show that they are extremal with respect to Kruskal's tensor rank bound.
\end{abstract}

\maketitle	



\section{Introduction}

The bilinear complexity of a given problem is defined as the minimum number of non-scalar multiplications required in any non-commutative algorithm solving the problem.
Each bilinear map is a 3-tensor, and its bilinear complexity is known to be its tensor rank. Computing the bilinear complexity of 3-tensors is NP-complete over any finite field and over the rationals is NP-hard \cite{halstad}. Much research has hence been focused on obtaining bounds on the tensor rank. 

Estimating the tensor rank of special bilinear maps, such as matrix multiplication, or multiplication in an algebra is a well studied problem that goes back some decades and continues to be an active area of research \cite{Alderstrassen, ballet+,brockett1978optimal, burgisser2013algebraic, kruskal1977three,sheekeylavrauw}.

The image of a 3-tensor is a linear space of matrices; especially in the case of the bilinear map being defined over a finite field, this image is a {\em matrix code}, often studied with respect to the {\em rank metric}. A lower bound on the tensor rank given by Kruskal 
in \cite{kruskal1977three} fits well with the standard parameters of a linear matrix code and leads to the concept of a {\em minimal tensor rank} (MTR) code, which is one that meets this bound with equality. 
Moreover, in the case of a matrix code, having low tensor rank is of interest since such codes perform well in terms of storage and encoding complexity. 
Every MTR matrix code implies the existence of a {\em maximum distance separable} (MDS) block code, while the converse is not yet answered. This problem has been considered in \cite{byrne2019tensor}, wherein classes of MTR codes were constructed.  In particular, it was noted there that
a class of one-dimensional Delsarte-Gabidulin codes comprises MTR codes and this was used to obtain an upper bound on the tensor rank of a class of $k$-dimensional Delsarte-Gabidulin codes. 

One characterization of the tensor rank of a linear space of matrices (a matrix code) is the minimum dimension of a {\em perfect space} that contains it. We call a space {\em perfect} if it has a basis of rank-one matrices. Then besides knowing the tensor rank of a matrix code, an interesting and challenging problem is the explicit construction of the basis of rank-one matrices (a {\em perfect base}) whose span contains it.
Even in the case of the 1-dimensional Delsarte-Gabidulin codes that are known to be MTR, perfect bases for this class of codes are not known in  general, although the question has been answered for $4 \times 4$ matrices over the finite fields of orders $2$ and $3$, as part of a classification of semi-fields \cite{sheekeylavrauw}.

In the case of {\em matrix pencils} (that is, linear spaces of matrices of dimension 2) perfect bases have been constructed \cite{jaja1979optimal}. In \cite{atkinson1983ranks}, the orthogonal complements of such spaces have been considered, however the techniques used by the authors do not extend beyond matrix pencils. We follow this line of research for a more general class of 3-tensors in $\K^k \otimes \K^n \otimes \K^m$ over an arbitrary field $\K$.
In our main result, we give an infinite family of 3-tensors that are perfect and construct perfect bases for this family. As an application of this result, we show that a class of the $(m-1)$-dimensional Delsarte-Gabidulin codes (which are duals of 1-dimensional Delsarte-Gabidulin codes) are MTR and construct perfect bases for these codes.

This paper is organised as follows. In Section 2 we introduce some preliminary notions. In Section 3 we construct perfect bases of a special class of 3-tensors, extending results of \cite{jaja1979optimal}. In Section 4 we consider the tensor rank of the orthogonal complements of the matrix spaces considered in Section 3. In Section 5 we specialize to matrix codes defined over finite fields, and especially the {\em vector rank-metric codes}. As a consequence of the results of Section 4, we show that certain $(m-1)$-dimensional Delsarte-Gabidulin codes meet Kruskal's bound with equality and we construct perfect bases for such codes.

\section{Preliminaries and Notation}

We fix some notation that will be used throughout the paper. 
For any integer $\ell$, we write $[\ell]$ to denote the set $\{1,\ldots,\ell\}.$
$\K$ denotes an arbitrary field, $q$ denotes a fixed prime power, and $\Fq$ denotes finite field of cardinality $q$. We write $\K^\times$ to denote the group of units of the field $\K$.
We let $n,m$ denote positive integers satisfying $2\leq n\leq m$. We denote by $\Fqnm$ the space of $n\times m$ matrices with entries in $\Fq$ and we let $E_{i,j}$ be the matrix in $\K^{n\times m}$ with a $1$ in position $(i,j)$ and $0$ elsewhere. Finally, we write $M$ to denote a matrix of the following form:	
	\begin{equation*}
		M:=\left(
		\begin{array}{c|c}
		\renewcommand{\arraystretch}{1}
			0 & I \\\hline
			 a_1 & a_{2} \;\;\cdots\;\; a_{m}
		\end{array}
		\right)\in\Kmm.
	\end{equation*}
That is, $M$ is the companion matrix of a polynomial $x^m-a_m x^{m-1}-\cdots a_2 x - a_1 \in \K[x]$.

We will study families of $3$-tensors in $\K^k\otimes\K^n\otimes\K^m$, with $1\leq k\leq nm$. 
We recall the main definitions and results in the following, and we refer the reader to \cite{burgisser2013algebraic} for further details. It is known that if $\{u_1,\ldots,u_k\}$, $\{v_1,\ldots,v_n\}$ and $\{w_1,\ldots,w_m\}$ are bases of $\K^k$, $\K^n$ and $\K^m$ respectively, then a basis for $\K^k\otimes\K^n\otimes\K^m$ is
	\begin{equation*}
		\{u_l\otimes v_i\otimes w_j: 1\leq \ell\leq k, 1\leq i\leq n, 1\leq j\leq m \}.
	\end{equation*}
	
In particular, we have $\dim(\K^k\otimes\K^n\otimes\K^m)=knm$. A $3$-tensor $X\in\K^k\otimes\K^n\otimes\K^m$ can be represented as a $3$-dimensional array, that is as a map 
	\begin{equation*}
		X:[k]\times [n] \times [m] \longrightarrow \K
	\end{equation*}
	given by $X=(X_{\ell ij}: 1\leq \ell\leq k, 1\leq i\leq n, 1\leq j\leq m)$. Therefore, the tensor 
	$\displaystyle X=\sum_{r=1}^R u_r\otimes v_r\otimes w_r$ is equivalently represented by the array
	\begin{equation*}
		X_{\ell ij}=\sum_{r=1}^R u_{\ell r}v_{ir}w_{jr}
	\end{equation*}
	where $u_r:=(u_{\ell r}: 1\leq \ell\leq k)$, $v_r:=(v_{ir}: 1\leq i\leq n)$ and $w_r:=(w_{jr}: 1\leq j\leq k)$. This representation of $X$ is called its \textbf{coordinate tensor}. We hence identify $\K^k\otimes\K^n\otimes\K^m$ with the space $\K^{k\times n\times m}$. 
	The coordinate tensor of $X$ is also represented as an array of matrices $X=(X_1\mid \ldots \mid X_k)$, $X_\ell= (X_{\ell ij}) \in \K^{n \times m}$.
		
\begin{definition}
	Let $X=(X_1\mid\cdots\mid X_k)\in\Kknm$. The \textbf{first slice space} of $X$ is denoted by $\ss_1(X)$ and is defined to be the span of $X_1,\ldots,X_k$ over $\K$. We say that $\ss_1(X)$ is \textbf{nondegenerate} (or alternatively that $X$ is $1$-\textbf{nondegenerate}) if $\dim(\ss_1(X))=k$.
\end{definition}

\begin{definition}
	Let $X\in \K^{k\times n\times m}$. We say that $X$ is a \textbf{simple} (or \textbf{rank}-1) \textbf{tensor} if there exist $u\in\K^k$, $v\in\K^n$ and $w\in\K^m$ such that $X=u\otimes v\otimes w$. 
\end{definition}
	Let $\displaystyle X = \sum_{r=1}^{R} u_r \otimes v_r \otimes w_r$. Consider map
\begin{align*}
	\mu:\K^k \times \Kknm \longrightarrow \Kknm : (a,X)&\longmapsto \mu(a,X)=\sum_{r=1}^R(a\cdot u_r)\otimes v_r\otimes w_r.
\end{align*}
This map yields a tensor of the form $\displaystyle \sum_{r=1}^R \lambda_r\otimes v_r\otimes w_r$ with $\lambda_r\in\K$, which can be identified with the matrix $\displaystyle \sum_{r=1}^R (\lambda_r v_r)\otimes w_r$ since $\K\otimes \K^n$ and $\K^n$ are isomorphic. 
The image of $\mu$ is $\ss_1(X)$ and in particular
\[ \ss_1(X) = \langle \mu(e_1,X),\ldots,\mu(e_k,X)\rangle_{\K} \]
where $e_s$ is the $s$-th element of the canonical basis for $\K^k$ for $1\leq s\leq k$.

\begin{definition}
	Let $V$ be a $k$-dimensional subspace of $\K^{n\times m}$ and let $X \in \K^{k \times n \times m}$ we say that $X$ is a {\bf generator tensor} of $V$ if $V=\ss_1(X)$.   
\end{definition}
	
The tensor rank of a 3-tensor, or its bilinear complexity, is defined as follows. 
	
\begin{definition}
	Let $X\in \Kknm$. We define the \textbf{tensor rank} of $X$, which we denote by $\trk(X)$, to be the least integer $R$ such that $X$ can be expressed as sum of $R$ simple tensors, i.e.
	\begin{equation*}
		\trk(X):=\min\left\{R\in \Z\;|\;\sum_{r=1}^R u_r\otimes v_r\otimes w_r \textup{ for some } u_r\in\K^k, v_r\in\K^n,w_r\in\K^m \right\}.
	\end{equation*} 
\end{definition}

In fact, tensor rank can be equivalently defined in respect of the minimum number of rank-one matrices whose span contains first slice space of a tensor. This prompts the following definitions.
	
\begin{definition}[{{\cite{atkinson1983ranks}}}]
	A $\K$-vector space of $n\times m$ matrices is called \textbf{perfect} if it is generated by a set of rank-$1$ matrices. 
\end{definition}	

\begin{definition}
    Let $X\in\Kknm$ and let $\A \subseteq\Knm$ be a set of linearly independent rank-$1$ matrices. We say that $\A$ is a \textbf{perfect base} of $X$ if the $\K$-span of $\A$ contains $\ss_1(X)$.
    If furthermore $\A$ has cardinality $R$, we say that $\A$ is an $R$-\textbf{base} of $X$. 
\end{definition}

The next result gives a characterization of the tensor rank (see \cite[Proposition~14.45]{burgisser2013algebraic} and \cite[Proposition~3.4]{byrne2019tensor}), from which it is immediate that the tensor rank of a tensor $X$ is the least dimension of any perfect space that contains $X$. 
	
\begin{lemma}\label{lem:bnrs19}
	Let $X\in\K^{k\times n\times m}$ and $R$ be a positive integer. The following are equivalent.
	\begin{enumerate}[leftmargin=5.5ex]
		\item $\trk(X)\leq R$.
		\item There exists an $R$-base for $X$.
		\item There exist diagonal matrices $D_1,\ldots,D_k\in\K^{R\times R}$ and matrices $P\in\K^{n\times R}$, $Q\in\K^{m\times R}$ such that $\ss_1(X)=P\<D_1,\ldots,D_k\> Q^T=\<PD_1Q^T,\ldots,PD_kQ^T\>$.
	\end{enumerate}
\end{lemma}

From Lemma \ref{lem:bnrs19}, it follows that if $X,X' \in \Kknm$ are both generator tensors of the same space $V \leq \Knm$ then $\trk(X)=\trk(X')$.
Therefore, we define the tensor rank of $V$ to be the tensor rank of any generator tensor for $V$. 
	
For any tensor $X=(X_1\mid \cdots\mid X_k)\in\Kknm$ and matrices $P\in\Knn,Q \in \Kmm$, we denote by $PXQ$ the tensor $(PX_1Q\mid \cdots\mid PX_kQ)$.

\begin{definition}
	The \textbf{dual} of a vector space $V\leq \Knm$ is defined to be its orthogonal complement with respect to the trace bilinear form: 
	\[V^\perp:=\{N\in\Knm:\Tr(MN^t)=0 \textup{ for all } M \in V, M\neq 0\},\] 
	where $\Tr(MN^t)$ is the trace of the square matrix $MN^t$.
\end{definition}

\section{Tensors of the form \texorpdfstring{$(I \mid M^{i_1}\mid \cdots \mid M^{i_\ell})$}{}}

In this section, we give explicit constructions of perfect bases of tensors in $\K^{k\times n\times m}$ of the form $(I \mid M^{i_1}\mid \cdots \mid M^{i_{k-1}})$, recalling that $M$ is a companion matrix of a polynomial in $\K[x]$. 
These results build on the work of \cite{jaja1979optimal}, in which tensors of the form $(I \mid M)$ are considered.

\begin{definition}
	We say that a pair of tensors $X,Y \in \Kknm$ are {\bf equivalent}
	if there exist $P\in\GL_n(\K)$ and $Q\in\GL_m(\K)$ such that 
	$\ss_1(X) = P\,\ss_1(Y)\,Q:=\{PMQ: M \in \ss_1(Y)\}$.
\end{definition}

\begin{remark}
\label{rem:PNP-1}
	Let $X,Y\in\Kknm$ be equivalent tensors and let $L\in\GL_n(\K)$ and $N\in\GL_m(\K)$ such that $Y=LXN$. It is easy to see that if $\A$ is a perfect base for $X$ then $\{LAN:A \in\A\}$ is a perfect base for $Y$.
\end{remark}
 
By Remark \ref{rem:PNP-1}, the results given in the following can be applied to the 3-tensor 
\[(LN\mid LMN\mid \ldots \mid LM^{k-1}N)\in\Kkmm,\] 
for $L,N\in\GL_m(\K)$. 

We recall the following lemma, which gives a lower bound on the tensor rank.
\begin{lemma}[{\cite[Lemma~3.1]{jaja1979optimal}}]\label{lem:lbtrk(I|M)}
	Let $A$ be an $m \times m$ matrix. Then $\trk(I\mid A) = m$ if and only if $A$ has $m$ distinct eigenvalues.
\end{lemma}

The following result was shown in \cite{jaja1979optimal}, wherein an explicit basis of a perfect space containing $\ss_1(I\mid M)$ was constructed, yielding an upper bound on $\trk(I\mid M)$. The lower bound is then given by Lemma \ref{lem:lbtrk(I|M)}.

\begin{theorem}[{\cite[Theorem~3.2]{jaja1979optimal}}]
\label{thm:jaja}
	Let $|\K|\geq m$. Then $\trk(I\mid M)=m$ if $M$ is diagonalizable and $\trk(I\mid M)=m+1$ otherwise.
\end{theorem}

In the following result, we show that the construction of a perfect base given in {\cite[Theorem~3.2]{jaja1979optimal}} (Theorem \ref{thm:jaja}) can be extended to the case of tensors of the form $(I\mid M\mid M^{-1})$ over any field of cardinality at least $m+1$.

\begin{theorem}
\label{thm:jajaext}
	Let $|\K|\geq m+1$ and let $f=(x-\alpha_1)\,\cdots\,(x-\alpha_{m-r})\,g\in \K[x]$ be the characteristic polynomial of $M$, where $g\in\K[x]$ has either degree $r\leq 1$ or is a polynomial of degree $r\geq 2$ that is not decomposable into linear factors. There exist $P\in\GL_m(\K)$ and $A,B\in\Kmm$ of rank $1$ such that the following hold.
	\begin{enumerate}[leftmargin=5.5ex]
		\item If $0\leq r\leq 1$ then $\{P^{-1}E_{i,i}P:i \in [m]\}$ is an $m$-base of  $(I\mid M\mid M^{-1})$.
		\item If $r=2$ then $\{P^{-1}E_{i,i}P:i \in [m]\}\cup\{A\}$ is an $(m+1)$-base of $(I\mid M\mid M^{-1})$.
		\item If $r\geq 3$ then $\{P^{-1}E_{i,i}P:i \in [m]\}\cup\{A,B\}$ is an $(m+2)$-base of $(I\mid M\mid M^{-1})$. 
	\end{enumerate}
\end{theorem}

\begin{proof}
	Suppose first that $r\in\{0,1\}$ and let $g(x)=x-\alpha_m$ in the instance that $r=1$. 
	Then both $M$ and $M^{-1}$ are diagonalizable and so there exists $P\in\GL_m(\K)$ such that
	$PMP^{-1}=\diag(\alpha_1,\ldots,\alpha_m)$ and $PM^{-1}P^{-1}=\diag(\alpha_1^{-1},\ldots,\alpha_m^{-1})$. This yields that
	\[ I = \sum_{i=1}^{m} P^{-1} E_{i,i} P,\qquad M=\sum_{i=1}^{m} \alpha_i P^{-1} E_{i,i} P,\qquad
	M^{-1}=\sum_{i=1}^{m} \alpha_i^{-1} P^{-1} E_{i,i} P.
	\]
	Since $P$ is invertible, $\rk(P^{-1}E_{i,i}P)=1$ for each $i\in  [m]$, which implies that 
	$\{P^{-1}E_{i,i}P:i \in [m]\}$ is an $m$-base for $(I\mid M \mid M^{-1})$. 
	
	Now suppose that $r\geq 2$. Let $\beta_1,\ldots,\beta_r\in\K\setminus\{0,\alpha_1,\ldots,\alpha_{m-r}\}$ be distinct and let $h:=(x-\beta_1)\,\cdots\,(x-\beta_r)$. Let $M_h$ be the companion matrix of $h$ and let 
		\begin{equation}\label{eq:D1}
			D_1=\left(
		\begin{array}{c|c}
		\renewcommand{\arraystretch}{1}
			0 & 0 \\\hline
			 0 & M_g-M_h
		\end{array}
		\right).
		\end{equation}
	
	Since $M_h$ is diagonalizable, there exists $Q_1\in\GL_r(\K)$ such that $Q_1 M_hQ_1^{-1}=\diag(\beta_1,\ldots,\beta_r)$.
    Moreover, as the $\beta_i$ are nonzero we have that $M_h$ is invertible. 
	
	Define the matrix 
	\begin{equation}\label{eq:Q}
		Q:=\left(
		\begin{array}{c|c}
		\renewcommand{\arraystretch}{1}
			I & 0 \\\hline
			 0 & Q_1\Tstrut{13pt}
		\end{array}
		\right)\in\Kmm.
		\end{equation}
	It is easy to see that
	\begin{equation*}
		Q\left(
		\begin{array}{c|c}
			\renewcommand{\arraystretch}{1}
			0 & 0 \\\hline
			0 & M_h\Tstrut{13pt}
		\end{array}
		\right) Q^{-1} =
		\left(
		\begin{array}{c|c}
			\renewcommand{\arraystretch}{1}
			0 & 0 \\\hline
			0 & \diag(\beta_1,\ldots,\beta_r) 
		\end{array}
	    \right)= 
		   \sum_{i=m-r+1}^m \beta_i E_{i,i}\Tstrut{13pt}.
	\end{equation*}
	There exists a matrix $P \in \GL_m(\K)$ such that:
	\begin{align*}
	PMP^{-1}&=\sum_{i=1}^{m-r}\alpha_i E_{i,i}+
	\left(\begin{array}{c|c}
		\renewcommand{\arraystretch}{1}
			0 & 0 \\\hline
			0 & M_g
		\end{array}
		\right)\\
	&=\sum_{i=1}^{m-r}\alpha_i E_{i,i}+
	\left(\begin{array}{c|c}
		\renewcommand{\arraystretch}{1}
			0 & 0 \\\hline
			0 & M_h
		\end{array}
		\right)+D_1\\
		&=\sum_{i=1}^{m-r}\alpha_i E_{i,i}+\sum_{i=m-r+1}^m \beta_{i-m+r}\;Q^{-1}E_{i,i}Q+D_1\\
		&=\sum_{i=1}^{m-r}\alpha_i \;Q^{-1}E_{i,i}Q+\sum_{i=m-r+1}^m \beta_{i-m+r}\,Q^{-1}E_{i,i}Q+D_1.
	\end{align*}
	
	We have  $\displaystyle I=\sum_{i=1}^m P^{-1}Q^{-1}E_{i,i}QP$ and from the above equation we get: 
	\begin{equation*}
	M=\sum_{i=1}^{m-r}\alpha_i \,P^{-1}Q^{-1}E_{i,i}QP+\sum_{i=m-r+1}^m \beta_{i-m+r}\,P^{-1}Q^{-1}E_{i,i}QP+P^{-1}D_1P.
	\end{equation*}
	Observe that if $r=2$ then $M_g^{-1}\in\<I,M_g\>$. Therefore, by Theorem \ref{thm:jajaext}, with $D_1$ and $Q$ defined as in (\ref{eq:D1}) and (\ref{eq:Q}) respectively, and the fact that
	\begin{equation*}
	    PM^{-1}P^{-1}-\sum_{i=1}^{m-2}\alpha_i^{-1}E_{i,i}=\left(
		\begin{array}{c|c}
		\renewcommand{\arraystretch}{1}
			 0 & 0 \\\hline
			 0 & M_g^{-1}\Tstrut{12pt}
		\end{array}
		\right),
	\end{equation*}
	we get $\ss_1((I\mid M\mid M^{-1}))\leq\<P^{-1}Q^{-1}E_{i,i}QP:i \in [m]\>+\<D_1\>$. Since $M$ is not diagonalizable we have $\trk(I\mid M\mid M^{-1})\geq m+1$, while equality follows from the existence of an $(m+1)$-base for
	$(I\mid M\mid M^{-1})$.
	
	Finally, if $r\geq 3$ we define
	\begin{equation}\label{eq:D2}
			D_2:=\left(
		\begin{array}{c|c}
		\renewcommand{\arraystretch}{1}
			0 & 0 \\\hline
			 0 & M_g^{-1}-M_h^{-1}\Tstrut{12pt}
		\end{array}
		\right).
	\end{equation}
	Once can check that 
	\begin{equation*}
		PM^{-1}P^{-1}=\sum_{i=1}^{m-r}\alpha_i^{-1} \;Q^{-1}E_{i,i}Q+\sum_{i=m-r+1}^m \beta_{i-m+r}^{-1}\,Q^{-1}E_{i,i}Q+D_2.
	\end{equation*}
	and therefore we get 
	\[\ss_1((I\mid M\mid M^{-1}))\leq\<P^{-1}Q^{-1}E_{i,i}QP:i \in \{1,\ldots,m\}\>+\<P^{-1} D_1 P,P^{-1}D_2 P\>.\qedhere\]
\end{proof}

We now extend the results of Theorem \ref{thm:jajaext} in the following.

\begin{corollary}
\label{cor:jajaext}
	Let $|\K|\geq m+1$ and let $f=(x-\alpha)\,\cdots\,(x-\alpha_{m-r})\,g\in \K[x]$ be the characteristic polynomial of $M$, where $g\in\K[x]$ has either degree $r\leq 1$ or is a polynomial of degree $r\geq 2$ that is not decomposable into linear factors. Let $S:=\{s_1,\ldots,s_\ell\}\subseteq\Z$ and 
	let $X:=(I\mid M\mid M^{-1}\mid M^{s_1}\mid \ldots \mid M^{s_\ell})$. There exist $P,Q\in\GL_m(\K)$ and $A,B\in\Kmm$ of rank $1$ such that the following hold. 
	\begin{enumerate}[leftmargin=5.5ex]
		\item If $0\leq r\leq 1$ then $\{P^{-1}Q^{-1}E_{i,i}QP:i\in[m]\}$ is an $m$-base of $X$.
		\item If $r=2$ then $\{P^{-1}Q^{-1}E_{i,i}QP:i\in[m] \}\cup\{A\}$ is an $(m+1)$-base of $X$.
		\item If $r= 3$ then $\{P^{-1}Q^{-1}E_{i,i}QP:i\in[m]\}\cup\{A,B\}$ is an $(m+2)$-base of $X$. 
	\end{enumerate}
\end{corollary}
\begin{proof}
	Throughout this proof we use the same notation as in the proof of Theorem \ref{thm:jajaext}. If $r\in\{0,1\}$ then any power of $M$ is diagonalizable and therefore by Theorem \ref{thm:jaja} we have $\ss_1((I\mid M\mid M^{-1}\mid M^{s_1}\mid \ldots \mid M^{s_\ell}))\leq \<P^{-1}E_{i,i}P:i\in[m]\>$. We thus now assume that $r\geq 2$. Observe that for every integer $s\in S$, we have
	\begin{equation*}
		PM^sP^{-1}=\left(
		\begin{array}{c|c}
		\renewcommand{\arraystretch}{1}
			\begin{matrix}
				\alpha_1^s & & \\
				& \ddots & \\
				& & \alpha_{n-r}^s
			\end{matrix} & 0 \\\hline
			 0 & M_g^s
		\end{array}
		\right).
	\end{equation*}
	
	One can easily check that if $r=2$ then for every $s\in S$ we have $M_g^s\in\<I,M_g\>$. Therefore, by Theorem \ref{thm:jajaext}, with $D_1$ and $Q$ defined as in (\ref{eq:D1}) and (\ref{eq:Q}) respectively, and the fact that
	\begin{equation*}
	    PM^sP^{-1}-\sum_{i=1}^{m-2}\alpha_i^sE_{i,i}=\left(
		\begin{array}{c|c}
		\renewcommand{\arraystretch}{1}
			 0 & 0 \\\hline
			 0 & M_g^s
		\end{array}
		\right),
	\end{equation*}
	we get 
	\begin{equation*}
		\ss_1((I\mid M\mid M^{-1}\mid M^{s_1}\mid \ldots \mid M^{s_\ell}))\leq\<P^{-1}Q^{-1}E_{i,i}QP:i\in [m]\>+\<P^{-1}D_1 P\>.
	\end{equation*}
    Similarly, if $r=3$ then for every $s\in S$ we have $M_g^s\in\<I,M_g,M_g^{-1}\>$. Again by Theorem \ref{thm:jajaext}, 
	with $D_1,Q$, and $D_2$ defined as in (\ref{eq:D1}), (\ref{eq:Q}), and (\ref{eq:D2}) respectively, and the fact that
	\begin{equation*}
	    PM^sP^{-1}-\sum_{i=1}^{m-3}\alpha_1^sE_{i,i}=\left(
		\begin{array}{c|c}
		\renewcommand{\arraystretch}{1}
			 0 & 0 \\\hline
			 0 & M_g^s
		\end{array}
		\right),
	\end{equation*}
	we have 
	\begin{equation*}
		\ss_1((I\mid M\mid M^{-1}\mid M^{s_1}\mid \ldots \mid M^{s_\ell}))\leq\<P^{-1}Q^{-1}E_{i,i}QP:i\in [m]\>+\<P^{-1}D_1 P,P^{-1}D_2 P\>.
	\end{equation*}
	This concludes the proof.
\end{proof}

\begin{notation}
\label{not:Y}
	We denote by $Y_n$ the $n\times m$ matrix $Y_n:=\left(\begin{array}{c|c}I & 0\end{array}\right)\in \Knm$. Clearly, for any matrix $P\in\K^{m\times m}$, $Y_nP$ is the matrix obtained by deleting the last $m-n$ rows from $P$.
\end{notation}

\begin{corollary}
\label{cor:n=2,3}
	Let $|\K|\geq m+1$, $n\in\{2,3\}$, and let $f=(x-\alpha)\,\cdots\,(x-\alpha_{m-r})\,g\in \K[x]$ be the characteristic polynomial of $M$, where $g\in\K[x]$ has either degree $r\leq 1$ or is a polynomial of degree $r\geq 2$ that is not decomposable into linear factors. 
	
	Let $X=\left(LY_nM^{-1}N\mid LY_nN\mid LY_nMN\mid \cdots \mid LY_nM^{m-2}N\right)\in\K^{m \times n \times m}$.
	For all $L\in\GL_n(\K)$ and $N\in\GL_m(\K)$, there exist $A,B\in\Kmm$ such that the following hold.
	\begin{enumerate}[leftmargin=5.5ex]
		\item If $0 \leq r \leq 1$ then $\{LY_nP^{-1}E_{i,i}PN:1 \leq m\}$ is a $m$-base of $X$.
		\item If $r=2$ then $\{LY_nP^{-1}E_{i,i}PN:1\leq i \leq m \}\cup\{LY_nAN\}$ is an $(m+1)$-base of $X$.
		\item If $r\geq 3$ then $\{LY_nP^{-1}E_{i,i}PN:1\leq i \leq m \}\cup\{LY_nAN,LY_nBN\}$ is an $(m+2)$-base of $X$. 
	\end{enumerate}
\end{corollary}
\begin{proof}
	Clearly, if $\A$ is a perfect base of $\left(Y_nM^{-1}\mid Y_n\mid Y_nM\mid \cdots \mid Y_nM^{m-2}\right)$ then $\{LAN:A \in \A\}$ is a perfect base of $X$, by Remark \ref{rem:PNP-1}. It therefore suffices to establish the result in the case that $L$ and $N$ are identity matrices, namely for the tensor $\left(Y_nM^{-1}\mid Y_n\mid Y_nM\mid \cdots \mid Y_nM^{m-2}\right)$. If $0\leq r\leq 3$ then the result follows from Corollary \ref{cor:jajaext}. We assume $r\geq 4$ in the remainder and we let $M_h$ be the companion matrix of $h:=(x-\alpha)\,\cdots\,(x-\alpha_{m-r})(x-\beta_1)\,\cdots\,(x-\beta_r)\in\K[x]$, with $\beta_1,\ldots,\beta_r\in\K\setminus\{0,\alpha_1,\ldots,\alpha_{m-r}\}$, and $P\in\GL_m(\K)$ be such that
	\begin{equation*}
		PM_hP^{-1}=\sum_{i=1}^{m-r}\alpha_iE_{i,i}+\sum_{i=m-r+1}^{m}\beta_{i-m+r}E_{i,i}.
	\end{equation*}
	Assume first that $n=2$. Then 
	\begin{equation*}
		\<Y_n,Y_nM,\ldots,Y_nM^{m-2}\>=\<Y_n,Y_nM_h,\ldots,Y_nM_h^{m-2}\>\leq\<Y_nP^{-1}E_{i,i}P:i\in [m]\>.
	\end{equation*}
	Let $D_1:=Y_n(M^{-1}-M_h^{-1})$, which has rank $1$. Clearly, we have
	\begin{equation*}
		\<Y_nM^{-1},Y_n,Y_nM,\ldots,Y_nM^{m-2}\>\leq\<Y_nP^{-1}E_{i,i}P:i\in [m]\>+\<D_1\>.
	\end{equation*}
	Finally, for $n=3$ we have 
	\begin{equation*}
		\<Y_n,Y_nM,\ldots,Y_nM^{m-3}\>=\<Y_n,Y_nM_h,\ldots,Y_nM_h^{m-3}\>\leq\<Y_nP^{-1}E_{i,i}P:i\in [m]\>.
	\end{equation*}
	Let $D_1:=Y_n(M^{-1}-M_h^{-1})$ and $D_2:=Y_n(M^{m-2}-M_h^{m-2})$ and notice that $\rk(D_1)=\rk(D_2)=1$. 
	Then
	\begin{equation*}
		\<Y_nM^{-1},Y_n,Y_nM,\ldots,Y_nM^{m-2}\>\leq\<Y_nP^{-1}E_{i,i}P:i \in [m]\>+\<D_1,D_2\>,
	\end{equation*}
	which concludes the proof.
\end{proof} 

\begin{remark}
\label{rem:m}
	One can observe that for $r=2$ or $n=2$, the result and the construction in the proof of Theorem \ref{thm:jajaext} and Corollaries \ref{cor:jajaext} and \ref{cor:n=2,3} hold for $|\K|\geq m$. Indeed $0$ can be a root of $h$ (or $f$) since in this case we do not need $M_h$ (or $M_f$) to be invertible.
\end{remark}

\begin{example}
    Let $m=5,n=3$, so that and $Y_n=(I_3\mid 0) \in\F_7^{3\times 5}$. Let $g:=x^3+6x^2+4\in\F_7[x]$, which is irreducible over $\F_7$, and let $f=(x-1)(x-2)g\in\F_7[x]$.
    The companion matrix of $f$ is given by
    \begin{equation*}
        M_f=\begin{pmatrix}
            0 & 1 & 0 & 0 & 0\\
            0 & 0 & 1 & 0 & 0\\
            0 & 0 & 0 & 1 & 0\\
            0 & 0 & 0 & 0 & 1\\
            6 & 5 & 5 & 2 & 4
        \end{pmatrix}
        \in\F_7^{5\times 5}.
    \end{equation*}
      We will use the construction given in the proof of Corollary \ref{cor:n=2,3} to find a $7$-base for the tensor $(Y_3M_f^{-1}\mid Y_3\mid Y_3M_f\mid Y_3M_f^2\mid Y_3M_f^3)$. Let $h=(x-1)(x-2)(x-3)(x-4)(x-5)\in\F_7[x]$, which has companion matrix:
    \begin{equation*}
        M_h=\begin{pmatrix}
            0 & 1 & 0 & 0 & 0\\
            0 & 0 & 1 & 0 & 0\\
            0 & 0 & 0 & 1 & 0\\
            0 & 0 & 0 & 0 & 1\\
            1 & 6 & 1 & 6 & 1
        \end{pmatrix}
        \in\F_7^{5\times 5}.
    \end{equation*}
     Let $P\in\GL_5(\F_7)$ be such that $PM_hP=\diag(5,4,3,2,1)$. First of all, observe that 
    \begin{equation*}
        \<Y_3,Y_3M_h,Y_3M_h^2\>=\<(I_3\mid 0),(0\mid I_3\mid 0), (0\mid I_3)\>=\<Y_3,Y_3M_f,Y_3M_f^2\>
    \end{equation*}
    and that $\<M_h^{-1},I,M_h,M_h^2,M_h^3\>\leq\<P^{-1}E_{i,i}P:i \in [5]\>$ since $M_h$ is diagonalizable. We define
    \begin{align*}
        D_1:=Y(M_f^{-1}-M_h^{-1})=
        \begin{pmatrix}
            4 & 6 & 1 & 5 & 5\\
            0 & 0 & 0 & 0 & 0\\
            0 & 0 & 0 & 0 & 0
        \end{pmatrix}\quad \textup{ and }\quad 
        D_2:=Y(M_f^3-M_h^3)=
        \begin{pmatrix}
            0 & 0 & 0 & 0 & 0\\
            0 & 0 & 0 & 0 & 0\\
            5 & 6 & 4 & 3 & 3
        \end{pmatrix}.
    \end{align*}
    Therefore, we have 
    \[Y_3M_f^{-1}=Y_3M_h^{-1}+D_1\in \<P^{-1}E_{i,i}P:i \in [5]\>+\<D_1\>,\] 
    and 
    \[ Y_3M_f^{3}=Y_3M_h^{3}+D_2\in \<P^{-1}E_{i,i}P:i \in [5]\>+\<D_2\>.\]
    This implies that a $7$-base for the tensor $(Y_3M_f^{-1}\mid Y_3\mid Y_3M_f\mid Y_3M_f^2\mid Y_3M_f^3)$ is \[\<Y_3P^{-1}E_{i,i}P:i \in [5]\>+\<D_1,D_2\>.\]
\end{example}

We conclude this section with an example that shows that in general for $r>3$, an $(m+r-1)$-base constructed as in Theorem \ref{thm:jajaext} cannot be enlarged to an $(m+r-1+j)$-base for a tensor of the of the form $(M^{-1}\mid I\mid M \mid M^2 \mid \ldots \mid M^{m-2})$, $M\in\GL_m(\K)$.

\begin{example}
	Let $f:=x^5+x+4, g:=(x-1)(x-2)(x-3)(x-4)(x-5)\in\F_7[x]$. Let
	\begin{equation*}
    M_f=\begin{pmatrix}
			0 & 1 & 0 & 0 & 0\\
			0 & 0 & 1 & 0 & 0\\
			0 & 0 & 0 & 1 & 0\\
			0 & 0 & 0 & 0 & 1\\
			3 & 6 & 0 & 0 & 0
		\end{pmatrix},\qquad 
		M_g=\begin{pmatrix}
			0 & 1 & 0 & 0 & 0\\
			0 & 0 & 1 & 0 & 0\\
			0 & 0 & 0 & 1 & 0\\
			0 & 0 & 0 & 0 & 1\\
			1 & 6 & 1 & 6 & 1
		\end{pmatrix},\qquad
		P:=\begin{pmatrix}
			3 & 6 & 0 & 4 & 1\\
			2 & 2 & 6 & 3 & 1\\
			5 & 6 & 0 & 2 & 1\\
			4 & 5 & 3 & 1 & 1\\
			1 & 0 & 1 & 0 & 1
		\end{pmatrix}
		\in\F_7^{5\times 5}.
	\end{equation*}
	$M_f$ and $M_g$ are the companion matrices of $f$ and $g$, respectively,
	while $P\in GL_5(\F_7)$ diagonalizes $M_g$, that is: 
	\begin{equation*}
		PM_gP^{-1}=\begin{pmatrix}
			5 & 0 & 0 & 0 & 0\\
			0 & 4 & 0 & 0 & 0\\
			0 & 0 & 3 & 0 & 0\\
			0 & 0 & 0 & 2 & 0\\
			0 & 0 & 0 & 0 & 1\\
		\end{pmatrix}.
	\end{equation*}
	Let $D_1:=M_f-M_g$ and $D_2:=M_f^{-1}-M_g^{-1}$, so that
	\begin{equation*}
		D_1:=\begin{pmatrix}
			0&0&0&0&0\\
			0&0&0&0&0\\
			0&0&0&0&0\\
			0&0&0&0&0\\
			2&0&6&1&6
		\end{pmatrix},\qquad
		D_2:=\begin{pmatrix}
			4&1&6&1&4\\
			0&0&0&0&0\\
			0&0&0&0&0\\
			0&0&0&0&0\\
			0&0&0&0&0
		\end{pmatrix}.
	\end{equation*}
	Let 
	\[\A:=\{P^{-1}E_{1,1}P,P^{-1}E_{2,2}P,P^{-1}E_{3,3}P,P^{-1}E_{4,4}P,P^{-1}E_{5,5}P\}\cup\{D_1,D_2\},\] 
	one can check that $M_f^{-1},I,M\in\<\A\>$ and there is no rank-1 matrix $N$ such that 
	$M^j \in \<\A\cup\{N\}\>$ for $j \in \{\pm 2,\pm 3\}$.
\end{example}

\section{Dual Tensors}

In \cite{atkinson1983ranks}, the authors applied Kronecker's theory of matrix pencils to show that the tensor-rank of any $1$-nondegenerate tensor in $\K^{(mn-2)\times n\times m}$ is $mn-2$ unless it is equivalent to an $(mn-2)\times n\times m$ tensor $X$ satisfying $X_{j,1,1}+X_{j,2,2}=0$ and $X_{j,1,2}=0$ for all $1\leq j\leq mn-2$. In other words, they showed that the first slice space of any such tensor is perfect. 
Inspired by this result, we show that some families of
$(mn-s)\times n\times m$ tensors are perfect and give explicit constructions
of $(mn-s)$-bases for such tensors for $s\geq 2$. We apply new methods to obtain these results, as Kronecker's theory no longer applies in this case.

Before proceeding, we remark that it is easy to see   that $\Knm$ is perfect. Moreover, it has been proved in \cite{atkinson1979maximal} that $mn-1$ is an upper bound for the tensor-rank of any $(mn-1)$-dimensional matrix space. By the {\em tensor-rank bound} (see Theorem \ref{prop:trkbound}) it follows that every subspace of $\Knm$ of dimension $mn-1$ is perfect.

\begin{remark}
	\label{rem:gendual}
	Let $X\leq\Knm$, $P\in\GL_n(\K)$, $Q\in\GL_m(\K)$, $L\in X$, $N\in X^\perp$ and define $PXQ:=\{PNQ:N \in X\}$. One can easily check that, by the properties of the trace, we have
	\begin{equation*}
		\Tr(PLQ((P^t)^{-1}N(Q^t)^{-1})^t)=\Tr(PLQQ^{-1}N^tP^{-1})=\Tr(PLN^tP^{-1})=\Tr(LN^t)=0.
	\end{equation*}
	This implies that the dual of $PXQ$ is $(P^{t})^{-1}X^\perp(Q^t)^{-1}$. 
   As a consequence, all the results in the remainder of this section can be easily adapted for equivalent tensors. In particular we have that, if $\A$ is a perfect base of $X^\perp$ then $(P^{t})^{-1}\A(Q^t)^{-1}$ is a perfect base of $(PXQ)^\perp$. 
\end{remark}

In the remainder we use the following notation.

\begin{notation}\label{not:JE}
	Let $\gamma\in\K^\times $. We denote by $J,\E(\gamma)\in\Kmm$ the matrices defined by
	\begin{equation*}
		J:=\left(
		\begin{array}{c|c}
		\renewcommand{\arraystretch}{1}
			0 & 1 \\\hline
		    I_{m-1} & 0
		\end{array}
		\right), \qquad 
		\E(\gamma):=\left(
		\begin{array}{ccccc}
			\gamma^{m-1} & \gamma^{m-1} & \cdots & \gamma & 1\\
			-\gamma^{m} & -\gamma^{m} & \cdots & -\gamma^2 & -\gamma\\\hline\\
			\multicolumn{5}{c}{\raisebox{\dimexpr\normalbaselineskip-0.2\ht\strutbox-.5\height}[0pt][0pt]{\bigzero}}
		\end{array}\right).
	\end{equation*}
	It is easy to see that multiplying a matrix by $J$ on the left (respectively, on the right) corresponds to the cyclic shift $(12\ldots m)$ of the rows downwards by one (respectively, to the cyclic shift $(12\ldots m)$ of the columns left by one). Also multiplication on the left (respectively, on the right) by $J^t$
	corresponds to the cyclic shift of the rows upwards (respectively to a cyclic shift of the columns to the right).
\end{notation}

\begin{lemma}
\label{lem:rotM}
    The following holds for all $i,j\in \{0,\ldots,m-1\}$.
    \begin{equation*}
        \left(J^j\right)^tM^{i-j}=
        \begin{cases}
            \left(\begin{array}{c|c}
                0_{(m-i)\times i} & I_{m-i} \\\hline
                * & * 
            \end{array}\right) & \textup{ if } i\geq j,\\\\
            \left(\begin{array}{c|c|c}
                0_{(m-j)\times i} & I_{m-j} & 0_{(m-j)\times (j-i)}   \\\hline
                * & * & * 
            \end{array}\right)  & \textup{ if } i<j.
        \end{cases}
    \end{equation*}
\end{lemma}
\begin{proof}
Recall that multiplying a matrix by $\left(J^j\right)^t$ on the left corresponds to the cyclic shift of the rows upwards by $j$. Let $i\geq j$ and notice that the matrix $M^{i-j}$ has the form 
    \begin{equation*}
        M^{i-j}=\left(\begin{array}{c|c}
                0_{(m-i+j)\times (i-j)} & I_{m-i+j} \\\hline
                * & * 
            \end{array}\right).
    \end{equation*}
    Multiplying $M^{i-j}$ on the left by $\left(J^j\right)^t$, results in an upwards cyclic shift of the rows of $M^{i-j}$ and so we have
    \begin{equation*}
        \left(J^j\right)^tM^{i-j}=
        \left(\begin{array}{c|c}
                0_{(m-i)\times i} & I_{m-i} \\\hline
                * & * 
            \end{array}\right).
    \end{equation*}
    On the other hand, let $i<j$ and note that the matrix $M^{i-j}$ has the form
    \begin{equation*}
        M^{i-j}=\left(\begin{array}{c|c}
            A & B  \\\hline
            I_{m-j+i} & 0_{(m-j+i)\times (j-i)} 
        \end{array}\right),
    \end{equation*}
    for some $A\in\K^{(j-i)\times (m-j+i)}$ and $B\in\GL_{j-i}(\K)$.
    Therefore, 
    \begin{equation*}
        \left(J^j\right)^tM^{i-j}=\left(\begin{array}{c|c|c}
                0_{(m-j)\times i} & I_{m-j} & 0_{(m-j)\times (j-i)}   \\\hline
                * & * & * 
            \end{array}\right).
    \end{equation*}
    This concludes the proof.
\end{proof}
	
The following is the main result of this section. 

\begin{theorem} \label{th:main}
\label{thm:base}
	Let $s\in [m-1]$, $|\K|\geq s+1$ and let $M\in\Kmm$ be invertible, i.e. $a_1\neq 0$ . The subspace $\<I,M,M^2,\ldots,M^{s-1}\>^\perp\leq\Kmm$ is perfect. In particular, an $(m^2-s)$-base for $\<I,M,M^2,\ldots,M^{s-1}\>^\perp$ is given by
	\begin{equation*}
		\A(\mS):=\,\{J^i\,E_{1,j}\,(M^{-i})^t:s+1\leq j\leq m,0\leq i\leq m-1\}\,\cup\,\{\,J^i\,\E(\gamma)\,(M^{-i})^t:0\leq i\leq m-2,\gamma\in\mS\}.
	\end{equation*}  
	where $\mS:=\{1,\gamma_1,\ldots,\gamma_{s-1}\}$ is a set of distinct elements of $\K^\times $.
\end{theorem}

\begin{definition}
    Let $A = (a_{ij}) \in \K^{n \times m}$.
    We define the {\bf vector representation} of $A$ to be the vector
    \[(a_{11},\ldots,a_{1m},a_{21},\ldots,a_{2m},\ldots,a_{n1},\ldots,a_{nm}) \in \K^{nm}.\]
\end{definition}

For ease of exposition, we first prove our main result for $m=4$ and $s=3$ in Section \ref{sec:ex} and we devote Section \ref{sec:proof} to giving a complete proof of Theorem \ref{thm:base}.

\subsection{Proof for \texorpdfstring{$m=4$}{} and \texorpdfstring{$s=3$}{}}
\label{sec:ex}

Throughout this section we let $m=4$ and $s=3$. Moreover, we assume $|\K|\geq 4$. Let $\mS:=\{1,\alpha,\beta\}$ be a set of distinct elements of $\K^\times $ and let $M\in\Kmm$ be invertible, i.e. $a_1\neq 0$. We will show that a $13$-base for $\<I,M,M^2\>^\perp$ is $\A:=\{A_i: i\in [13]\}$, where
	\begin{equation*}
	\scalemath{0.76}{
	\begin{array}{llll}
		A_1:=\left(\begin {array}{rrrr} 
		1&1&1&1\\ 
		-1&-1&-1&-1\\
		0&0&0&0\\ 
		0&0&0&0\\
		\end {array}
		\right),&
		A_2:=\left(\begin {array}{rrrr} 
		\alpha^3 &\alpha^2 &\alpha &1\\ 
		-\alpha^4 &-\alpha^3 &-\alpha^2 &-\alpha \\
		0&0&0&0\\ 
		0&0&0&0\\
		\end {array}
		\right),&
		A_3:=\left(\begin {array}{rrrr} 
		\beta^3 &\beta^2 &\beta &1\\ 
		-\beta^4 &-\beta^3 &-\beta^2 &-\beta \\
		0&0&0&0\\ 
		0&0&0&0\\
		\end {array}
		\right),&
		A_4:=\left(\begin {array}{rrrr} 
		0&0&0&1\\ 
		0&0&0&0\\
		0&0&0&0\\ 
		0&0&0&0\\
		\end {array}
		\right),
	\end{array}
	}
	\end{equation*}
	and the remaining matrices of $\A$ are 
	\begin{equation*}
	\renewcommand\arraystretch{1.4}
	\begin{array}{*{4}{L{3.5cm}}}
		A_5:=JA_1(M^{-1})^t, & A_6:=JA_2(M^{-1})^t, &
		A_7:=JA_3(M^{-1})^t, & A_8:=JA_4(M^{-1})^t,\\
		A_9:=J^2A_1(M^{-2})^t,& A_{10}:=J^2A_2(M^{-2})^t, &
		A_{11}:=J^2A_3(M^{-2})^t, & A_{12}:=J^2A_4(M^{-2})^t,\\
		A_{13}:=J^3A_4(M^{-3})^t.
	\end{array}
	\end{equation*}
	Clearly, $A_1,A_2,A_3,A_4$ all have rank 1. Moreover, they are trace-orthogonal to each of the matrices $I,M,M^2$; that is, $\Tr(A_j)=\Tr(A_jM^t)=\Tr(A_j(M^2)^t)=0$ for $j \in [4]$.
	Since $J$ and $M$ are both invertible, it follows that $\rk(A_j)=1$ for each $j \in [13]$.
	Let us examine the first 2 rows of $(J^i)^tM^rM^{-i}=(J^i)^tM^{r-i}$.
	As a consequence of Lemma \ref{lem:rotM}, for any $i,r\in\{0,1,2\}$ we have
	    \begin{equation*}
	    (J^i)^tM^{r-i}=\left(
	        \begin{array}{c|c|c}
	             0_{2\times r} & I_2 & 0_{2\times (m-r-2)}  \\\hline
	             * & * &* 
	        \end{array}\right)\qquad \textup{and} \qquad
	        (J^{3})^tM^{r-3}=\left(
	        \begin{array}{c|c|c}
	             0_{1\times r} & 1 & 0_{1\times (3-r)}  \\\hline
	             * & * &* 
	        \end{array}\right).
	    \end{equation*}
    Therefore, for each $r\in\{0,1,2\}$, we have
    \begin{align*}
    \begin{array}{lll}
         (J^i)^tM^{r-i}(A_1)^t=\left(
         \renewcommand{\arraystretch}{1.2}
            \begin{array}{c|c}
                \begin{array}{cc}1 & -1\\1&-1\end{array}& 0  \\\hline
                * & 0 
            \end{array}
         \right) && \textup{ for all } i\in\{0,1,2\},\\\\
         (J^i)^tM^{r-i}(A_2)^t=\left(
         \renewcommand{\arraystretch}{1.2}
            \begin{array}{c|c}
                \begin{array}{cc}\alpha^{m-r-1} & -\alpha^{m-r}\\\alpha^{m-r-2}&-\alpha^{m-r-1}\end{array}& 0  \\\hline
                * & 0 
            \end{array}
         \right) &\qquad\qquad & \textup{ for all } i\in\{0,1,2\},\\\\
         (J^i)^tM^{r-i}(A_3)^t=\left(
         \renewcommand{\arraystretch}{1.2}
            \begin{array}{c|c}
                \begin{array}{cc}\beta^{m-r-1} & -\beta^{m-r}\\\beta^{m-r-2}&-\beta^{m-r-1}\end{array}& 0  \\\hline
                * & 0 
            \end{array}
         \right) && \textup{ for all } i\in\{0,1,2\},\\\\
         (J^i)^tM^{r-i}(A_4)^t=\left(
         \renewcommand{\arraystretch}{1.2}
            \begin{array}{c|c}
               0& 0  \\\hline
                * & 0 
            \end{array}
         \right)
         && \textup{ for all } i\in\{0,1,2,3\}.
    \end{array}
    \end{align*}
	
	It follows that 
	\begin{itemize}[leftmargin=5.5ex]
	    \item $\Tr(J^i A_j(M^{-i})^t (M^r)^t)=\Tr((J^i)^tM^{r-i}(A_j)^t)=0$ for each $j \in [3]$ and $i,r\in\{0,1,2\}$;
	    \item $\Tr(J^i A_4(M^{-i})^t (M^r)^t)=\Tr((J^i)^tM^{r-i}(A_4)^t)=0$ for each $r\in\{0,1,2\}$ and $i\in\{0,1,2,3\}$.
	\end{itemize}
	Hence, we have $\<\A\> \leq \<I,M,M^2\>^\perp$. It remains to show that $\A$ is a set of $13$ linearly independent matrices. Consider the following disjoint subsets of $\A$. 
	\begin{align*}
		\A_{0}:=\{A_1,A_2,A_3,A_4\}, &&& \A_{1}:=\{A_5,A_6,A_7,A_8\}, &&& \A_{2}:=\{A_9,A_{10},A_{11},A_{12}\}, &&& \A_{3}:=\{A_{13}\}.
	\end{align*}
	We define the following matrices $B_0,\ldots,B_4$, whose rows are comprised of vector representations of selected submatrices of the members of $\A$.
	Let $B_0$ be the $4 \times 8$ matrix whose $i$-th row is the vector representation of the $2 \times 4$ submatrix of $A_i$ comprising its first 2 rows (which are both non-zero for $i\in [3]$).
	\begin{equation*}
	\begin{aligned}
		B_0&:=\left(\begin{array}{c|c}
			\begin{array}{cccc}
				1&1&1&1\\
				\alpha^3 &\alpha^2 &\alpha &1\\
				\beta^3 &\beta^2 &\beta &1\\
				0&0&0&1\\   
			\end{array}&
			\begin{array}{cccc}
				-1&-1&-1&-1\\
				-\alpha^4 &-\alpha^3 &-\alpha^2 &-\alpha\\
				-\beta^4 &-\beta^3 &-\beta^2 &-\beta\\
				0&0&0&0\\   
			\end{array}
		\end{array}\right):=\left(\begin{array}{c|c}B_0^{(1)}&B_0^{(2)}\end{array}\right),\\[5pt]
	\end{aligned}
	\end{equation*}
	Now let $B_0^{(1)}$ and $B_0^{(1)}$ be the $4 \times 4$ submatrices of $B_0$ comprising its 4 leftmost and 4 rightmost columns respectively. Let $\overline{B_0}^{(1)} = (0 \quad0\quad0\quad1)$, which is the  
	last row of $B_0^{(1)}$. Now define $B_1,B_2,B_3 $ as follows.
	\begin{equation*}
	\begin{aligned}
		B_1&:=\left(\begin{array}{c|c}B_0^{(1)}\left(M^{-1}\right)^t&B_0^{(2)}\left(M^{-1}\right)^t\end{array}\right),\\[5pt]
		B_2&:=\left(\begin{array}{c|c}B_0^{(1)}\left(M^{-2}\right)^t&B_0^{(2)}\left(M^{-2}\right)^t\end{array}\right),\\[5pt]
		B_3&:=\left(\begin{array}{c}\overline{B_0}^{(1)}\left(M^{-3}\right)^t\end{array}\right).
		\end{aligned}
	\end{equation*}
	
	Note that each $i$-th row of $B_1$ is the vector representation
	of the $2 \times 4$ submatrix of $A_{4+i}$ comprising its 2nd and 3rd rows, which are non-zero for $i\in [3]$, etc.
	
	Observe that the $3\times 3$ submatrix of $B_0^{(1)}$ given by the first $3$ rows and columns is a Vandermonde matrix, since $1,\alpha,\beta$ are distinct in $\K$; hence $\rk(B_0^{(1)})=4$. Since $M\in\GL_m(\K)$, we therefore have $\rk(B_0)=\rk(B_1)=\rk(B_2)=4$ and $\rk(B_1)=1$. 
	Finally, let $B$ be the $13\times 16$ matrix defined as follows:
	\begin{equation*}
		\renewcommand\arraystretch{2}
 		B=\left(\begin{array}{*{4}{C{2.3cm}}}
		\cblue B_0^{(1)} & \temppp{B_0^{(2)}} & \temppp{\phantom{\hspace{1.5cm}}} &	 \bigzero
		\\\hhline{-|-|-|}
		& \temppp{\cblue B_0^{(1)}\left(M^{-1}\right)^t} & \tempp{B_0^{(2)}\left(M^{-1}\right)^t}
		\\\hhline{~-|-|-|}
		& & \tempp{\cblue B_0^{(1)}\left(M^{-2}\right)^t} & B_0^{(2)}\left(M^{-2}\right)^t
		\\\hhline{~~-|-|}
		\bigzero & & \temp{} & \cor\overline{B_0}^{(1)}\left(M^{-3}\right)^t
	\end{array}\right).
	\end{equation*}
	It is now apparent that
	\begin{equation*}
		\rk(B)=\rk\left(B_0^{(1)}\right)+\rk\left(B_0^{(1)}\left(M^{-1}\right)^t\right)+\rk\left(B_0^{(1)}\left(M^{-2}\right)^t\right)+\rk\left(\overline{B_0}^{(1)}\left(M^{-3}\right)^t\right)=4+4+4+1=13.
	\end{equation*}
	This implies that $\A$ is a set of $13$ linearly independent rank-$1$ matrices that are trace-orthogonal to $\{I,M,M^2\}$ and therefore $\A$ is a $13$-base for $\<I,M,M^2\>^\perp$.

\subsection{Proof of Theorem \ref{thm:base}}
\label{sec:proof}

We first establish some notation and give some preparatory lemmas. 
Recall that $s$ is an integer satisfying $s\in [m-1]$, that we assume $|\K|\geq s+1$, that $\mS:=\{1,\gamma_1,\ldots,\gamma_{s-1}\}$ is a set of distinct elements of $\K^\times$ and that $M\in\Kmm$ in invertible.
Let $J$ and $\E(\gamma)$ be defined as in Notation \ref{not:JE}. As defined in the statement of Theorem \ref{thm:base}, recall that
\begin{equation*}
		\A(\mS):=\,\{J^i\,E_{1,j}\,(M^{-i})^t:s+1\leq j\leq m,\,0\leq i\leq m-1\}\,\cup\,\{\,J^i\,\E(\gamma)\,(M^{-i})^t:0\leq i\leq m-2,\gamma\in\mS\}.
	\end{equation*}
\begin{notation}

\label{not:A}
	For each integer $i \in \{0,\ldots,m-2\}$, we define  the sets
	\begin{align*}
		\A_i(\mS)&:=\{J^i\,E_{1,j}\,(M^{-i})^t: s+1\leq j\leq m\}\cup\{J^i\,\E(\gamma)\,(M^{-i})^t:s\in\mS\},\\
		\A_{m-1}(\mS)&:=\{J^{m-1}\,E_{1,j}\,(M^{1-m})^t: s+1\leq j\leq m\}.
	\end{align*}
	The $\A_0,\ldots,\A_{m-1}$ are pairwise disjoint and their union is $\A$. For ease of notation, in the following we write $\A$ and $\A_i$ instead of $\A(\mS)$ and $\A_i(\mS)$ for all $0\leq i\leq m-1$.
\end{notation}

\begin{lemma}
\label{lem:orthog}
	If $P\in\A$ then $\rk(P)=1$ and $P\in\<I,M,\ldots,M^{s-1}\>^\perp$.
\end{lemma}
\begin{proof}
	Clearly, for each $ j \in \{s+1,\ldots, m\}$ and $\gamma\in\mS$, we have $\rk(E_{1,j})=\rk(\E(\gamma))=1$. It can be checked that $E_{1,j},\E(\gamma)\in\<I,M,\ldots,M^{s-1}\>^\perp$. Since $J^i$ and $M^i$ are invertible for each $i$, each member of 
	$\A$ has rank equal to one.
	Finally, as a consequence of Lemma \ref{lem:rotM} for any $i\in\{0,\ldots,m-1\}$ and $r \in \{0,\ldots,s-1\}$, we have
	    \begin{equation*}
	    (J^j)^tM^{r-j}=\left(
	        \begin{array}{c|c|c}
	             0_{2 \times r} & I_2 & 0_{m-r-2}  \\\hline
	             * & * &* 
	        \end{array}\right)\quad \textup{ and } \quad
	        (J^{m-1})^tM^{r+1-m}=\left(
	        \begin{array}{c|c|c}
	              0_{1\times r} & 1 & 0_{1\times (m-1-r)}  \\\hline
	             * & * &* 
	        \end{array}\right).
	    \end{equation*}
	
	Therefore, for each $r\in\{0,\ldots, s-1\}$, we have
	\begin{equation}\label{eq:JME1j}
	    (J^i)^tM^{r-i}(E_{1,j})^t=\left(
         \renewcommand{\arraystretch}{1.2}
            \begin{array}{c|c}
               0& 0_{1\times (m-1)}  \\\hline
                * & 0_{(m-1)\times (m-1)} 
            \end{array}
         \right),
	\end{equation}
	for all $i\in\{0,\ldots,m-1\}$ and $j\in\{s+1,\ldots,m\}$, and 
	\begin{equation}\label{eq:JMEp}
        (J^i)^tM^{r-i}(\E(\gamma))^t=\left(
         \renewcommand{\arraystretch}{1.2}
            \begin{array}{c|c}
                \begin{array}{cc}\gamma^{m-r-1} & -\gamma^{m-r}\\\gamma^{m-r-2}&-\gamma^{m-r-1}\end{array}& 0_{2\times(m-2)}  \\\hline
                * & 0_{(m-2)\times (m-2)} 
            \end{array}
         \right),
	\end{equation}
	for all $i\in\{0,\ldots,m-2\}$ and $\gamma\in\mS$. Therefore, for each $r\in\{0,\ldots, s-1\}$, we get the following. 
	\[
	\Tr(M^r(J^j E_{1,j} (M^{-i})^t)^t)=\Tr(M^{r-i} (E_{1,j})^t (J^j)^t)=\Tr((J^j)^t M^{r-i} (E_{1,j})^t)=0,
	\]
	for all $i\in\{0,\ldots,m-1\}$ and $j\in\{s+1,\ldots,m\}$ and
	\[
	\Tr(M^r(J^i \E(\gamma) (M^{-i})^t)^t)=\Tr(M^{r-i} \E(\gamma)^t (J^i)^t ) = 0,
	\]
	for all $i\in\{0,\ldots,m-2\}$ and $\gamma\in\mS$. Therefore, $\<\A\> \leq \<I,M,\ldots,M^{s-1}\>^\perp$.
\end{proof}

\begin{lemma}
\label{lem:A0}
	$\A_0$ is a set of linearly independent matrices.
\end{lemma}
\begin{proof}
	Let $B_0$ be the $m\times 2m$ matrix whose rows are the vector representations of the nonzero rows (the first two rows) of the matrices in $\A_0$, i.e.
	\begin{equation}
		\label{eq:B0}
	\scalemath{0.95}{
		B_0:=\left(\renewcommand\arraystretch{1.5}
		\begin{array}{c|c|c}
			\cblue L_1 & L_2 & L_3\\\hline
			0 & \cor I_{m-s} & 0
		\end{array}\right)=
		\left(
		\begin{array}{c|c|c}
			\cblue \begin{array}{ccc}
				1 & \cdots & 1\\
				\gamma_1^{m-1} & \cdots & \gamma_1^{m-s}\\
				\vdots & &\vdots \\
				\gamma_{s-1}^{m-1} &  \cdots & \gamma_{s-1}^{m-s}
			\end{array}&
			\begin{array}{ccc}
				1 & \cdots & 1\\
				\gamma_1^{m-s-1} & \cdots &  1\\
				\vdots & & \vdots\\
				\gamma_{s-1}^{m-s-1} & \cdots &  1
			\end{array}&
			\begin{array}{ccc}
				1 & \cdots & -1\\
				-\gamma_1^{m} &  \cdots & -\gamma_1\\
				\vdots & &\vdots\\
				-\gamma_{s-1}^{m} &  \cdots & -\gamma_{s-1}
			\end{array}\Bstrut{-27pt}\\\hline
			\bigzero &
			\cor \begin{array}{C{1.2cm}cc}
				1 & &\\
				& \ddots & \\
				& & 1
			\end{array}&
			\bigzero\Bstrut{27pt}
		\end{array}\right).
	}
	\end{equation}
	Note that $L_1\in\K^{s\times s}$, $L_2\in\K^{s\times(m-s)}$, $L_3\in\K^{s\times m}$. 
	In particular, $L_1$ is row-equivalent to a full-rank Vandermonde matrix since the $1,\gamma_1,\ldots,\gamma_{s-1}$ are distinct in $\K^\times$. Therefore, we have $\rk(B_0)=\rk(L_1)+\rk(I_{m-s})=s+m-s=m$ which implies the statement.
\end{proof}

We are now ready to give a proof of Theorem \ref{thm:base}.

\begin{proof}[Proof of Theorem \ref{thm:base}]
	Lemma \ref{lem:orthog} shows that all the matrices in $\A$ have rank $1$ and are trace-orthogonal to $M^i$ for all $0\leq i\leq s-1$. It remains to show that $\A$ is a set of $m^2-s$ linearly independent matrices. Let $B_0,L_1,L_2,L_3$ be defined as in \eqref{eq:B0} and define matrices $B_0^{(1)},B_0^{(2)}\in\K^{m\times m}$ and $\overline{B_0}^{(1)}\in\K^{(m-s)\times m}$ by
	\begin{equation*}
		B_0^{(1)}:=\left(\begin{array}{c|c}
			L_1 & L_2\\\hline
			0 & I_{m-s}
		\end{array}\right),
		\qquad
		B_0^{(2)}:=\left(\begin{array}{c}
			L_3\\\hline
			0
		\end{array}\right),
		\qquad
		\overline{B_0}^{(1)}:=\left(\begin{array}{c|c}
			0 & I_{m-s}
		\end{array}\right).
	\end{equation*}
	For $0\leq i\leq m-2$, let $B_i$ be the $m\times 2m$ (resp. $B_{m-1}$ be the $(m-s)\times m$) matrix whose rows are the vector representation of the $(i+1)$-th and the $(i+2)$-th, (resp. the last row) of the matrices in $\A_i$ (resp. $\A_{m-1}$). These are the non-zero rows of the matrices $A_i$. 
	Since $\A_{4i+p} = J^i A_{p}(M^{-i})^t$ for $0 \leq i \leq m-2, p \in [m]$, we have
	\begin{itemize}[leftmargin=5.5ex]
		\item $\displaystyle B_i=\left(\begin{array}{c|c}
			B_0^{(1)}(M^{-i})^t & B_0^{(2)}(M^{-i})^t
		\end{array}\right)$ for all $0\leq i\leq m-2$ and
		\item $\displaystyle B_{m-1}=\left(\begin{array}{c}
			\overline{B_0}^{(1)}(M^{1-m})^t 
		\end{array}\right)$.
	\end{itemize}
	Finally, let $B$ be the $(m^2-s)\times m^2$ matrix defined by:
	\begin{equation*}
	\label{eq:B}
		B=\left(
		\renewcommand\arraystretch{2}
		\begin{array}{*{5}{C{1.5cm}}}
				\cblue B_0^{(1)} & \temppp{B_0^{(2)}} &\temppp{} & &\scalebox{1.5}{\bigzero}\\\hhline{-|-|-|}
				& \temppp{\cblue B_0^{(1)}(M^{-1})^t} & \temppp{B_0^{(2)}(M^{-1})^t} & \temppp{} &\\\hhline{~-|-|-|}
				&& \temppp{\cblue\rotatebox{20}{$\ddots$}} & \temppp{\rotatebox{20}{$\ddots$}} &\temppp{}\\\hhline{~~-|-|-|}
				& & & \temppp{\cblue B_0^{(1)}(M^{2-m})^t} & \temppp{B_0^{(2)}(M^{2-m})^t}\\\hhline{~~~-|-|}
				\raisebox{\dimexpr\normalbaselineskip-0.5\ht\strutbox}[0pt][0pt]{\scalebox{1.5}\bigzero}& & & & \temppp{\cor \overline{B_0}^{(1)}(M^{1-m})^t }
 			\end{array}\right).
	\end{equation*}
	The rank of $B$ is given by the sum of the ranks of the blocks on the main diagonal, that is
	\begin{equation*}
		\rk(B)=\rk\left(\overline{B_0}^{(1)}(M^{1-m})^t\right)+\sum_{i=0}^{m-2}\rk\left(B_0^{(1)}(M^{-i})^t\right)=m-s+(m-1)m=m^2-s.
	\end{equation*}
	The result now follows.
\end{proof}

\subsection{Consequences}
In this section we discuss some implications of Theorem \ref{thm:base}. In particular, we derive some new results and we show that some known results of complexity theory arise as corollaries of Theorem \ref{thm:base}. We start with a generalization of our results for $n\times m$ rectangular matrices.

\begin{corollary}
	Let $1\leq s\leq m-1$, $|\K|\geq s+1$, $B\in\GL_m(\Fq)$ and $\mS:=\{1,\gamma_1,\ldots,\gamma_{s-1}\}$ be a set of distinct elements of $\K^\times $. Then $\<B^{-1},B^{-1}M,\ldots,B^{-1}M^{s-1}\>^\perp\leq \Knm$ is a perfect space. 
	In particular, an $(nm-s)$-base for $\<B^{-1},B^{-1}M,\ldots,B^{-1}M^{s-1}\>^\perp$ is given by
	\begin{equation*}
		\overline{\A}:=\{B^tJ^i\,E_{1,j}\,(M^{-i})^t:s+1\leq j\leq m,0\leq i\leq n-1\}\,\cup\,\{\,B^tJ^i\,\E(\gamma)\,(M^{-i})^t:0\leq i\leq n-2,\gamma\in\mS\}.
	\end{equation*}
\end{corollary}
\begin{proof}
    Let $V = \<I,M,\ldots,M^{s-1}\>^\perp$. By Remark \ref{rem:gendual}, we have
    \[
    \<B^{-1},B^{-1}M,\ldots,B^{-1}M^{s-1}\>^\perp = (\{ B^{-1}T : T \in V \})^\perp =
    \{ B^t U : U \in V^\perp \}.
    \]
    The result is thus an immediate consequence of Theorem \ref{thm:base}.
\end{proof}

\begin{corollary}
\label{cor:basenm}
	Let $1\leq s\leq m-1$, $|\K|\geq s+1$. Then $\<Y_n,Y_nM,\ldots,Y_nM^{s-1}\>^\perp\leq \Knm$ is a perfect space. 
	In particular, an $(nm-s)$-base for $\<Y_n,Y_nM,\ldots,Y_nM^{s-1}\>^\perp$ is given by
	\begin{equation*}
		\overline{\A}:=\{Y_nJ^i\,E_{1,j}\,\left(M^{-i}\right)^t:s+1\leq j\leq m,0\leq i\leq n-1\}\,\cup\,\{\,Y_nJ^i\,\E(\gamma)\,\left(M^{-i}\right)^t:0\leq i\leq n-2,\gamma\in\mS\}.
	\end{equation*}
	where $\mS:=\{1,\gamma_1,\ldots,\gamma_{s-1}\}$ is a set of distinct elements of $\K^\times$.
\end{corollary}
\begin{proof}
	Theorem \ref{thm:base} implies that
	\begin{equation*}
		\A:=\{J^i\,E_{1,j}\,\left(M^{-i}\right)^t:s+1\leq j\leq m,0\leq i\leq m-1\}\,\cup\,\{\,J^i\,\E(\gamma)\,\left(M^{-i}\right)^t:0\leq i\leq m-2,\gamma\in\mS\}
	\end{equation*}
	is an $(m^2-s)$-base for the space $\<I,M,\ldots,M^{s-1}\>^\perp$ and, as we already observed, multiplying by $Y_n$ is equivalent to deleting the last $m-n$ rows. For each $i\in \{0,\ldots, m-1\}$, let $\A_i$ be the subset of $\A$ defined as in Notation \ref{not:A}. 
	For each $i\in \{0,\ldots,n-2\}$, define
	\[ 
	 \overline\A_i:=\{Y_n A : A \in \A_i\}, \overline\A_{n-1}:=\{Y_n(J^{n-1})^t\,E_{1,j}\,\left(M^{n-1}\right)^t:s+1\leq j\leq m\}
	 \text{ and }
		\overline\A:=\bigcup_{i=0}^{n-1}\overline\A_i.
	\]
	Let $i \in \{0,\ldots,n-1\}, j \in \{s+1\}$ and consider $J^i E_{1,j} (M^{-i})^t  = E_{i+1,j} (M^{-i})^t$. 
	Clearly this matrix has the $j$-th row of $(M^{-1})^t$ as its $(i+1)$-th row, and has all other rows equal to zero. In particular, since $i\leq n-1$, its last $m-n$ rows are zero. Similarly,
	The last $m-n$ rows of $J^i\,\E(\gamma)\,\left(M^{-i}\right)^t$ are all-zero for 
	$i \in \{0,\ldots,n-2\}, \gamma \in \mS$. Therefore,
	the matrices in $\overline\A$ correspond to matrices in $\A$ whose last $m-n$ rows are zero. 
	Therefore, for any matrix $P \in \bar{\A}$, we have $\Tr(Y_n M^r P^t) = \Tr(M^r P^t) = 0$, by Theorem  
	\ref{thm:base}. In other words, the matrices in $\overline\A$ have rank $1$ and are trace orthogonal to all $Y_n,Y_nM,\ldots,Y_nM^{s-1}$. 
	It remains to show that $\overline\A$ has cardinality $nm-s$. Since the sets $\overline\A_i$'s are disjoint, we have
	\begin{equation*}
		|\A|=\left|\bigcup_{i=0}^{n-1}\overline\A_i\right|=\sum_{i=0}^{n-1}\left|\overline{\A}_i\right|=n(m-2)+s(n-1)=nm-s,
	\end{equation*}
	which implies the statement.
\end{proof}

\begin{remark}\label{rem:cornmbase}
	Note that for $i \in\{n,\ldots,m-1\}$ and each $P\in\A_i$ we have $Y_nP=0$. Moreover, for all $\gamma\in\mS$ we have 
	\begin{equation*}
		\Tr\left(Y_nM^{n-1}\left(Y_n(J^{n-1})^t\E(\gamma)(M^{1-n})^t\right)^t\right)=\Tr(Y_n\E(\gamma)J^{n-1}Y_n^t)\neq 0.
	\end{equation*}
	Therefore, these matrices cannot be in $\<Y_n,Y_nM,\ldots,Y_nM^{s-1}\>^\perp$.
\end{remark}

In \cite[Lemma~1]{atkinson1983ranks}, the authors proved that $\<Y_n,Y_nM\>^\perp\leq\K^{n\times (n+1)}$ is perfect and they gave an explicit construction of an $(n^2-n+1)$-base for this space. Although not explicitly stated, their result holds for $\ch(\K)\neq2$. We retrieve this result as a consequence of Corollary \ref{cor:basenm} for a certain choice of the parameters and we show that our construction is independent of the characteristic of the field.

\begin{proposition}[{{\cite[Lemma~1]{atkinson1983ranks}}}]
\label{prop:atk}
	Let $\ch(\K)\neq 2$. The space $\<Y_n,Y_nM\>^\perp\leq\K^{n\times (n+1)}$ is perfect and an $(n^2+n-2)$-base for it is given by the set
	\begin{align*}
		\A:=&\{Y_nE_{i,j}:i\in [n],j\in [n+1],j\neq i, j\neq i+1 \}\\
		&\cup\{Y_n(E_{i,i}+E_{i,i+1}+E_{i,i+2}-E_{i+1,i}-E_{i+1,i+1}-E_{i+1,i+2}):i\in [n-1]\}\\
		&\cup\{Y_n(E_{i,i}-E_{i,i+1}+E_{i,i+2}+E_{i+1,i}-E_{i+1,i+1}+E_{i+1,i+2}):i\in [n-1]\}.
	\end{align*}
\end{proposition}

\begin{remark}
Observe that if $\ch(\K)= 2$ then 
	\begin{equation*}
		E_{i,i}+E_{i,i+1}+E_{i,i+2}-E_{i+1,i}-E_{i+1,i+1}-E_{i+1,i+2}=E_{i,i}-E_{i,i+1}+E_{i,i+2}+E_{i+1,i}-E_{i+1,i+1}+E_{i+1,i+2},
	\end{equation*}
	for all $1\leq i\leq n-1$, and so $|\A|=n^2-1$. As a consequence, the construction for the base and therefore the proof of \cite[Lemma~1]{atkinson1983ranks} fails. Corollary \ref{cor:basenm} shows that the space $\<Y_n,Y_nM\>^\perp\leq\K^{n\times (n+1)}$ is perfect for fields of arbitrary characteristic, provided
	the field has at least 3 elements.
	
Moreover, for $\ch(\K)\neq 2$, the construction of the perfect base in \cite[Lemma~1]{atkinson1983ranks} is almost a special case of Corollary \ref{cor:basenm} for $m=n+1$, with the choice $\mS=\{1,-1\}$. 
Explicitly, it is straightforward to check that 
	\[\{Y_nE_{i,j}:1\leq i\leq n,1\leq j\leq n+1,j\neq i, j\neq i+1 \}=\{Y_nJ^i\,E_{1,j}\,(M^{-1})^t:1\leq i\leq n, 3\leq j\leq n+1\},\]
		and that for all $i\in [n-1]$, we have
	\begin{align*}
	    E_{i,i}+E_{i,i+1}+E_{i,i+2}-E_{i+1,i}-E_{i+1,i+1}-E_{i+1,i+2}=J^{i+1}\,\left(\E(1)-\sum_{j=4}^{n+1}E_{1,j}\right)\,\left(M^{-i-1}\right)^t,\\
		E_{i,i}-E_{i,i+1}+E_{i,i+2}+E_{i+1,i}-E_{i+1,i+1}+E_{i+1,i+2}=(-1)^{n-i}J^{i+1}\,\left(\E(1)-\sum_{j=4}^{n+1}E_{1,j}\right)\,\left(M^{-i-1}\right)^t.
	\end{align*}	
	Of course $|\mS|=2$ only if $\ch(\K)\neq 2$.
\end{remark}

\begin{remark}
\label{rem:baseM}
    Recall that the bottom row of $M$ is assumed to be $(a_1,a_2,\ldots,a_m)$.
	Let $1\leq s\leq m-1$, $|\K|\geq s+1$, $a_1\neq 0$ and let $\mS:=\{1,\gamma_1,\ldots,\gamma_{s-1}\}$ be a set of distinct elements of $\K^\times$. Then $\<Y_n,Y_nM,\ldots,Y_nM^{s-1}\>^\perp\leq\Knm$ is perfect. In particular an $(nm-s)$-base for $\<Y_n,Y_nM,\ldots,Y_nM^{s-1}\>^\perp$ is given by
	\begin{equation*}
		\A:=\,\{Y_nJ^i\,E_{1,j}\,(M^{-i})^t:s+1\leq j\leq m,0\leq i\leq n-1\}\,\cup\,\{\,Y_nJ^i\,\E(\gamma)\,(M^{-i})^t:0\leq i\leq n-2,\gamma\in\mS\}.
	\end{equation*}  
\end{remark}

\begin{lemma}
\label{lem:2x2}
	Let $M\in\K^{2\times 2}$. Then $\<I,M\>^\perp$ is perfect if one of the following conditions holds.
	\begin{enumerate}[label={(\arabic*)},leftmargin=5.5ex]
		\item \label{item1:2x2} $a_1=0$ and $a_2\neq 0$.
		\item \label{item2:2x2} $a_1\neq 0$ and the polynomial $a_1x^2+a_2x-1\in\K[x]$ has two distinct roots in $\K^\times$.
	\end{enumerate}
\end{lemma}
\begin{proof}
    
	First note that if $a_1=0$ and $a_2\neq 0$ then then it is easy to check that the following is a $2$-base for $\<I,M\>^\perp$,
	\begin{equation*}
		\A:=\left\{\renewcommand\arraystretch{1.2}
		\begin{pmatrix}
			0 & 0 \\ 1 & 0 
		\end{pmatrix}, 
		\begin{pmatrix}
			a_2^{-1} & 1\\ -a_2^{-2} & -a_2^{-1}
		\end{pmatrix}\right\}.
	\end{equation*}
	This proves that if $a_1=0$ and $a_2\neq 0$ then $\<I,M\>^\perp$ is perfect. 
	Now assume that $a_1\neq 0$. Let $\gamma \in \K^\times$. We have 
	\[
	 \Tr(\E(\gamma)M^t) = 
	 \Tr\left( \begin{pmatrix}
			\gamma & 1 \\ -\gamma^2 & -\gamma 
		\end{pmatrix}
		\begin{pmatrix}
			0 & a_1 \\ 1 & a_2 
		\end{pmatrix}
	\right)	= -a_1\gamma^2 - a_2 \gamma +1
	\]
	 $\Tr(\E(\gamma)M^t)=0$ and $\rk(\E(\gamma)=1$ if and only if $\gamma$ is a root of $a_1x^2+a_2x-1\in\K[x]$. Clearly, $\E(\gamma)$, having zero trace, is orthogonal to $I$. 
	 Therefore, if the polynomial $a_1x^2+a_2x-1\in\K[x]$ has two distinct roots in $\K^\times$, say $\gamma_1,\gamma_2$, then $\A:=\{\E(\gamma_1),\E(\gamma_2)\}$ is a $2$-base for $\<I,M\>^\perp$.
\end{proof}

The following result shows that for any $s\in[m-1]$ the space $\<I,M,\ldots,M^{s-1}\>^\perp$ is perfect even when $M$ is not invertible.

\begin{corollary}
\label{cor:a1=0}
	Let $i:=\min(\{1\leq j\leq m:a_j\neq 0\})$. The following hold.
	\begin{enumerate}[label={(\arabic*)},leftmargin=5.5ex]
		\item \label{item1:subM} Let $1\leq s\leq m-i\leq m-1$ and let $|\K|\geq s+1$. 
		Then $\<I,M,\ldots,M^{s-1}\>^\perp$ is perfect.
		\item \label{item2:subM} If $i=m-1$ and the polynomial $a_{m-1}x^2+a_mx-1\in\K[x]$ has two distinct roots in $\K^\times $ then $\<I,M\>^\perp$ is perfect.
		\item If $i=m$ then $\<I\>^\perp$ and $\<I,M\>^\perp$ are perfect.
	\end{enumerate}
\end{corollary}

\begin{proof}
	For the case $i=1$, the result is given by Theorem \ref{thm:base}. Assume $2\leq i\leq m-1$ and observe that $M$ has the form
	\begin{equation*}
		M=\left(
		\begin{array}{*{6}{c}}
			\cblue 0 & \cblue I_{i-2} & \cblue 0 & \cblue 0 & \cblue \cdots & \cblue 0\\
			\cblue 0 & \cblue 0 & \cblue 1 & \cblue 0 & \cblue \cdots & \cblue 0\\\hhline{~~-|-|-|-}
			\cblue 0 & \cblue 0 & \temppp{\cgr 0} & \cgr 1 & \cgr  & \cgr 0\\\hline
			& &\temppp{\cor }  &\cor & \cor\ddots & \cor\\
			\multicolumn{2}{c}{\bigzero} &\temppp{\cor} &\cor &\cor &\cor 1\\
			& & \temppp{\cor a_i} &\cor a_{i+1} &\cor \cdots &\cor a_m
		\end{array}
		\right),
	\end{equation*}
	where the bottom right block is a full-rank $(m-i+1)\times (m-i+1)$ companion matrix, which we will call $\overline M$.	Let $\overline\A_1$ be the $(im-s)$-base for $\<Y_i,Y_iJ^t,\ldots,Y_i(J^t)^{s-1}\>^\perp$, $Y_i\in\K^{i\times m}$ and let $\overline\A_2$ be the $((m-i+1)^2-s)$-base for $\<I_{m-i+1},\overline M,\ldots, \overline M^{s-1}\>^\perp$, given by Corollary \ref{cor:basenm} for $M=J^t$ and $M=\overline M$ respectively,  and define
	\begin{equation*}
		\A_1:=\left\{\begin{pmatrix}A\\\hline0\end{pmatrix}: A\in\overline\A_1\right\} \qquad \textup{and} \qquad \A_2:=\left\{\left(\begin{array}{c|c}0 & 0\\\hline 0&A\end{array}\right): A\in\overline\A_2\right\}.
	\end{equation*}
	We claim that an $(m^2-s)$-base for  
	$\<I,M,\ldots,M^{s-1}\>^\perp$ is
	\begin{equation*}
		\A:=\A_1\cup\A_2\cup\{E_{k,j}:i+1\leq k\leq m, 1\leq j\leq i-1\}. 
	\end{equation*}
	The following are straightforward observations.
	\begin{itemize}[leftmargin=5.5ex]
		\item $\A_1\cap\{E_{k,j}:i+1\leq k\leq m, 1\leq j\leq i-1\}=\emptyset$.
		\item $\A_2\cap\{E_{k,j}:i+1\leq k\leq m, 1\leq j\leq i-1\}=\emptyset$.
		\item $\A_1\cap\A_2=\{E_{i,j}:i+s\leq j\leq m \}$.
	\end{itemize}
	Clearly, $\A_1$ and $\A_2$ are linearly independent, by construction and so by the above observations, $\A$ is linearly independent.
	Moreover $\<\A \> \leq \<I,M,\ldots,M^{s-1}\>^\perp$.
	It remains to show that $\A$ has cardinality $m^2-s$. We have
	\begin{align*}
		|\A|&=|\A_1\cup\A_2\cup\{E_{k,j}:i\leq k\leq m, 1\leq j\leq i-1\}|\\
		&=|\{E_{k,j}:i+1\leq k\leq m, 1\leq j\leq i-1\}|+|\A_1|+|\A_2|-|\{E_{i,j}:i+s\leq j\leq m \}|\\
		&=(i-1)(m-i)+im-s+(m-i+1)^2-s-(m-i-s+1)\\
		&=m^2-s,
	\end{align*}
	which implies \ref{item1:subM}. The remaining part of the statement follows from Lemma \ref{lem:2x2} using the same technique as in the first part of the proof.
\end{proof}

\section{The tensor rank of \texorpdfstring{$\Fqm$}{}-linear codes}

In the previous sections we constructed perfect bases of some families of $3$-tensors, which gives an upper bound on the tensor rank. In this section we first give an upper-bound on the tensor-rank for some families of tensors whose first slice space is $\Fqm$-linear. We then study the tensor-rank of $\Fqm$-linear \textit{rank-metric codes} with a particular focus on the $\Fqm$-linear \textit{generalized twisted Gabidulin codes}. We will conclude this section by establishing the existence of a new family of codes that attain the tensor-rank bound. In the remainder we use the following notation. 

\begin{notation}
\label{not:gamma}
	Let $\Gamma:=\{\gamma_1,\ldots,\gamma_m\}$ be a basis of $\Fqm$ over $\Fq$.  
	For each $i \in [n]$ and $\theta \in \Fqm$ we define the vector $\Gamma(\theta)\in\Fq^m$ by
	\begin{equation*}
	\theta=\sum_{i=1}^m\Gamma(\theta)_{j}\gamma_j.
	\end{equation*}
	For each $v\in\Fqm^n$ we denote by
	$\Gamma(v)$ the $n\times m$ matrix in $\Fqnm$ whose $i$-th row is $\Gamma(v_i)$. 
	For an $\Fqm$-subspace $V$ of $\Fqm^n$, we define the $\Fq$-subspace $\Gamma(V):=\<\Gamma(v):v \in V\>_{\Fq} \subseteq \Fqnm$.
\end{notation}

Clearly, the map $v\longmapsto \Gamma(v)$ is an $\Fq$-isomorphism and $\Gamma(V)$ is an $\Fq$-linear space.

\begin{remark}\label{rem:trkgab1}
Let $\alpha$ be a primitive element of $\Fqm$ and let $\Gamma=\{1,\alpha,\ldots,\alpha^{m-1}\}$.
    Let $M$ be the companion matrix of the minimal polynomial of $\alpha$ over $\Fq$.
    With respect to the basis $\Gamma$ for each $s \in \{1,\ldots,m-1\}$ and $i \in \{0,\ldots,m-1\}$ we have 
    \[\Gamma((1,\alpha,\ldots,\alpha^{s-1})\alpha^i ) = Y_s M^i .\]
	Then for each $s \in \{1,\ldots,m-1\}$, we have 
	\[\Gamma\left(\<(1,\alpha,\ldots,\alpha^{s-1})\>_{\Fqm}\right) = \< Y_s, Y_sM, Y_sM^2,\ldots,Y_sM^{m-1}\>_{\Fq}.\]
	 Clearly the rank of every non-zero element of $\Gamma\left(\<(1,\alpha,\ldots,\alpha^{s-1})\>_{\Fqm}\right)$
	 is $s$, since every non-trivial $\Fq$-linear combination of the $M^i$ is invertible.
\end{remark}

We now introduce some definitions and well-known facts on \textit{rank-metric codes} and we study the tensor-rank of some families of $\Fqm$-linear \textit{rank-metric codes} of which the above mentioned vector spaces are examples. The reader is referred to the survey \cite{gorla2019rank} for further details. 

\begin{definition}
	A (\textbf{rank-metric matrix}) \textbf{code} is a subspace $\C\leq\Fqnm$. The \textbf{minimum} (\textbf{rank}) \textbf{distance} of a non-zero code $\C$ is 
	\[d(\C):=\min\{\rk(X):X \in \C, X\neq 0\}\] 
	and for $\C:=\{0\}$, we define $d(\C)$ to be $n+1$. 
\end{definition}

From now on, unless otherwise stated, $\C\leq \Fqnm$ is a rank-metric code whose dimension is denoted by $k$ and its minimum distance by $d$. We refer to $n,m,k,d$ as the parameters of the code. 
Since every matrix code $\C$ is the first slice space of a tensor in $\Fqknm$, the $\trk(\C)$ is well defined (see \cite{byrne2019tensor} for a detailed treatment).

In \cite[Corollary~1]{kruskal1977three}, Kruskal gave a lower bound on the tensor-rank of a $1$-nondegenerate tensor. In the next result, we recall a coding-theoretical formulation of this bound given in \cite{byrne2019tensor}. 

\begin{theorem}[Tensor-rank bound]
\label{prop:trkbound}
 	We have $\trk(\C)\geq \dimq(\C)+d(\C)-1$.
\end{theorem}

We say that $\C$ is \textbf{MTR} (\textbf{minimal tensor-rank}) is $\C$ meets the bound in Theorem \ref{prop:trkbound}. 

\begin{definition}
	A \textbf{vector} (\textbf{rank-metric}) \textbf{code} is an subspace $C\leq \Fqm^n$.
\end{definition}

It is possible to obtain a matrix rank-metric code from a vector rank-metric code exploiting the fact that $\Fqm^n$ and $\Fqnm$ are isomorphic as $\Fq$-spaces. 
For a fixed basis $\Gamma$ of $\Fqm^n$ over $\Fq$ an explicit isomorphism is given by $v \mapsto \Gamma(v)$, as described in Notation \ref{not:gamma}.

If $C\leq\Fqm^n$ is a vector rank-metric code then the minimum distance of $\Gamma(C)$ does not depend on the choice of the basis $\Gamma$ of $\Fqm/\Fq$. Moreover, for any such basis we have
	\begin{equation*}
		\dimq(\Gamma(C))=m\cdot\dimqm(C).
	\end{equation*}

We define the \textbf{minimum distance} of a vector code $C\leq\Fqm^n$ to be the minimum distance of $\Gamma(C)$ for any choice of a basis $\Gamma$ of $\Fqm/\Fq$.
Likewise, we define the {\bf tensor rank of the vector code} $C$ to be $\trk(C) := \trk(\Gamma(C))$.

It is well-known that for any $s \in [n]$, $m+s-1$ is an upper bound on the tensor rank of the $\Fqm$-vector space $\<\left(1,\alpha,\ldots,\alpha^{s-1}\right)\>_{\Fqm}$
    for $q \geq m+s-2$  (see \cite[Proposition~14.47]{burgisser2013algebraic} 
    and \cite[Corollary~5.14]{byrne2019tensor}). 
    Moreover, from Kruskal's tensor rank bound, it can be seen that 
    \[\trk\left(\<\left(1,\alpha,\ldots,\alpha^{s-1}\right)\>_{\Fqm}\right) \geq m+s-1,\] 
    and so we have equality.
    We summarize this as follows.
    
\begin{lemma}
\label{lem:trkgab1}
    Let $\alpha$ be a primitive element of $\Fqm$ and let $s\in[n]$ be a positive integer satisfying $q \geq m+s-2$. 
        The tensor rank of $\<(1,\alpha,\ldots,\alpha^{s-1})\>_{\Fqm}$ is exactly $m+s-1$.
\end{lemma}

\begin{lemma}
\label{lem:lindep}
    Let $s$ be a positive integer such that $1\leq s \leq n$.
    Let $\beta_1,\ldots,\beta_n\in\Fqm$ such that
    $\<\beta_1,\ldots,\beta_n\>_{\Fqm}$ has dimension $s$ and $\<\beta_1,\ldots,\beta_n\>_{\Fqm}=\<\beta_1,\ldots,\beta_s\>_{\Fqm}$.
    Then \[\trk\left(\<(\beta_1,\beta_2,\ldots,\beta_n)\>_{\Fqm}\right) =
	\trk\left(\<(\beta_1,\beta_2,\ldots,\beta_s)\>_{\Fqm}\right).\]
\end{lemma}
\begin{proof}
    Let $\alpha$ be a primitive element of $\Fqm$ and let $\Gamma=\{1,\alpha,\ldots,\alpha^{m-1}\}$.
    Let $M$ be the companion matrix of the minimal polynomial of $\alpha$ over $\Fq$.
    We may assume that $\{\beta_{1},\ldots,\beta_{s}\}$ is an $\Fq$-basis of $\<\beta_1,\ldots,\beta_n\>_{\Fq}$. Let $\lambda^{(i)}_j \in \Fq$ be the uniquely determined coefficients of $\beta_i$ with respect to the basis $\{\beta_1,\ldots,\beta_s\}$. That is, 
	\begin{equation}\label{eq:betai}
	     \beta_i=\sum_{j=1}^s\lambda_j^{(i)}\beta_j, \: s+1\leq i\leq n.
	\end{equation}
	
	Let $N\in\Fq^{n\times m}$ be the matrix whose $i$-th row is $\Gamma(\beta_i)$ 
	for each $i\in [n].$
	Then
	\[V:=\Gamma(\<(\beta_1,\beta_2,\ldots,\beta_n)\>_{\Fqm}) = \< N, NM, NM^2,\ldots,NM^{m-1}\>_{\Fq}\]
	Let $L\in V$. 
	Then $L = \sum_{j=0}^{m-1} \theta_j NM^{j-1}$ for some
	$\theta_j \in \Fq$.
    For each $n \times m$ matrix $P$ over $\Fq$ we denote by $\overline{P}$ the $s \times m$ submatrix of $P$ comprising its first $s$ rows, i.e. $\overline{P}:=Y_sP$.  
	Then $\overline{N} = \Gamma(\beta_1,\ldots,\beta_s)$ and  
	\[\overline{L}=  \overline{N}\sum_{j=0}^{m-1} \theta_j M^{j-1} \in \< \overline{N}, \overline{N}M, \overline{N}M^2,\ldots,\overline{N}M^{m-1}\>_{\Fq} =  \Gamma(\<(\beta_1,\beta_2,\ldots,\beta_s)\>_{\Fqm}).\]
	Now $L_i = \sum_{j=0}^{m-1} \theta_j N_iM^{j-1} = \sum_{j=0}^{m-1} \theta_j \Gamma(\beta_i)M^{j-1}$ for each $i \in [n]$ and so
	the $i$-th row of $L$ for $s+1\leq i\leq n$ is a linear combination of the rows of $\overline{L}$. Explicitly, from (\ref{eq:betai}) we have
	\begin{equation*}
		L_i=\sum_{j=1}^s\lambda_j^{(i)}L_j .
	\end{equation*}
	
	Let $\overline\A\subseteq\Fq^{s\times m}$ be perfect base of
	$\overline{V}:=\Gamma(\<(\beta_1,\beta_2,\ldots,\beta_s)\>_{\Fqm})$.
	We claim that a perfect base of $V$ is given by 
	\begin{equation*}
	\A:=\left\{\left(\renewcommand\arraystretch{1.4}
	\begin{array}{c}
		\overline A\\\hline
		\sum_{j=1}^s\lambda_j^{(s+1)}\overline A_j\\
		\vdots\\
		\sum_{j=1}^s\lambda_j^{(n)}\overline A_j
	\end{array}\right):\overline A\in\overline\A \right\}\subseteq\Fqnm.
	\end{equation*}
	By construction, for each
	$A \in \A$ we have $\overline{A} \in \overline{\A}$ and so $\rk(A) = \rk(\overline{A})=1$ as each $i$-th row of $A$ is an $\Fq$-linear combination of the rows of $\overline{A}$ for $s+1\leq i\leq n$.
	 
	Moreover, $\A$ is a set of $|\overline{\A}|$ linearly independent matrices by the linear independence of $\overline\A$. 
	
	It remains to prove that $V\leq\<\A\>$. Let $\A=\{A^{(1)},\ldots, A^{(\ell)}\}$.
	Since $\overline{A}$ is a perfect base of $\overline{V}$, there exist  $\gamma_1,\ldots,\gamma_{\ell}\in\Fq$ such that
	$\bar{L} = \sum_{j=1}^\ell \gamma_j \overline{A^{(j)}}$. Then for $s+1 \leq i \leq n$ we have 
	\begin{equation*}
		L_i=\sum_{j=1}^s\lambda_j^{(i)}L_j  
		= \sum_{j=1}^s\lambda_j^{(i)} \sum_{k=1}^\ell \gamma_k \overline{A^{(k)}}_i 
		= \sum_{k=1}^\ell \gamma_k \sum_{j=1}^s\lambda_j^{(i)} \overline{A^{(k)}}_i
		= \sum_{k=1}^\ell \gamma_k A^{(k)}_i.
	\end{equation*} 
	Therefore, $L = \sum_{j=1}^\ell \gamma_j A^{(j)}$ and so $L \in \< \A \>$ and so
	$\A$ is a perfect base of $V$.
    In particular, if $|\overline{A}| = \trk(\overline{V})$, then $\A$ as constructed above is a perfect base of $V$ satisfying $\trk(V) \leq \trk(\overline{V})$. Clearly $\trk(\overline{V})\leq \trk (V)$ and so the result follows. 
\end{proof}

\begin{example}
	Let $f:=x^4+4x^2+4x+2\in\F_5[x]$, let $\alpha$ a root of $f$ and let $M$ be the companion matrix of $f$.
	The polynomial $f$ is irreducible over $\F_5$ and
	$\Gamma:=\{1,\alpha,\alpha^2,\alpha^3\}$ is a basis of $\F_{5^4}=\F_5[\alpha]$ over $\F_5$.
	Let $\overline N\in\F_5^{3\times 4}$ be the matrix whose $i$-th row is $\Gamma(\alpha^i)$, i.e.
	\begin{equation*}
		\overline N:=
		\begin{pmatrix}
			1 & 0 & 0 & 0\\
			0 & 1 & 0 & 0\\
			0 & 0 & 1 & 0
		\end{pmatrix}.
	\end{equation*} 
	
	From Lemma \ref{lem:trkgab1}, we know that the tensor rank of $\<(1,\alpha,\alpha^2)\>_{\F_{5^4}}$ is $6$.
	We will construct a $6$-base for $\<N,NM,NM^2,NM^3\>_{\F_5}$, where $N_i:=\Gamma(\alpha^i)$, $0\leq i\leq 2$, and $N_4:=\Gamma(4+3\alpha+2\alpha^2)$, i.e.
	\begin{equation*}
		N:=
		\begin{pmatrix}
			1 & 0 & 0 & 0\\
			0 & 1 & 0 & 0\\
			0 & 0 & 1 & 0\\
			4 & 3 & 2 & 0
		\end{pmatrix}.
	\end{equation*}
	Consider the $6$-base $\overline\A$ of $\<\overline N,\overline NM,\overline NM^2,\overline NM^3\>_{\F_5}$ obtained using the construction outlined in the proof of Corollary \ref{cor:n=2,3}, namely: 
	\begin{equation*}
	\scalemath{0.96}{
		\overline\A:=\left\{
		\begin{pmatrix}
			4&2&0&1\\3&4&0&2\\1&3&0&4
		\end{pmatrix},
		\begin{pmatrix}
			1&1&0&1\\3&3&0&3\\4&4&0&4
		\end{pmatrix},
		\begin{pmatrix}
			3&0&2&4\\2&0&3&1\\3&0&2&4
		\end{pmatrix},
		\begin{pmatrix}
			2&3&3&4\\2&3&3&4\\2&3&3&4
		\end{pmatrix},
		\begin{pmatrix}
			3&3&4&0\\0&0&0&0\\0&0&0&0
		\end{pmatrix},
		\begin{pmatrix}
			0&0&0&0\\0&0&0&0\\0&2&1&1
		\end{pmatrix}
		\right\}.}
	\end{equation*}
	In the notation of Lemma \ref{lem:lindep}, we have $m=4,n=s=3$. 
	As in the proof of Lemma \ref{lem:lindep}, we define the set $\A$ as follows:
	\begin{equation*}
		\A:=\left\{\renewcommand\arraystretch{1.4}
		\left(
		\begin{array}{c}
			\overline A\\\hline
			4\overline A_1+3\overline A_2+2\overline A_2
		\end{array}\right): \overline A \in \overline \A
		\right\},
	\end{equation*}
	which is explicitly given by:
	\begin{equation*}
		\scalemath{0.96}{
		\A=\left\{
		\begin{pmatrix}
			4&2&0&1\\3&4&0&2\\1&3&0&4\\2&1&0&3
		\end{pmatrix},
		\begin{pmatrix}
			1&1&0&1\\3&3&0&3\\4&4&0&4\\1&1&0&1
		\end{pmatrix},
		\begin{pmatrix}
			3&0&2&4\\2&0&3&1\\3&0&2&4\\4&0&1&2
		\end{pmatrix},
		\begin{pmatrix}
			2&3&3&4\\2&3&3&4\\2&3&3&4\\3&2&2&1
		\end{pmatrix},
		\begin{pmatrix}
			3&3&4&0\\0&0&0&0\\0&0&0&0\\2&1&1&0
		\end{pmatrix},
		\begin{pmatrix}
			0&0&0&0\\0&0&0&0\\0&2&1&1\\0&4&2&2
		\end{pmatrix}
		\right\}.}
	\end{equation*}
	One can check that $\A$ is indeed a $6$-base for  $\<N,NM, NM^2,NM^3\>_{\F_5}$ and further, the tensor rank is $6$.
\end{example}

One class of $\Fqm$-linear matrix spaces of particular interest in coding theory is the class of {\em Delsarte-Gabidulin} codes and their generalisations, such as the {\em twisted Gabidulin codes}. These codes are extremal with respect to the following (rank-metric Singleton) bound.

\begin{proposition}[{\cite[Theorem~5.4]{delsarte1978bilinear}}]
\label{prop:SingletonRM}
	We have $k\leq m(n-d+1)$.
\end{proposition}

We say that $\C$ is an \textbf{MRD} (\textbf{maximum rank-distance}) code if $\C$ satisfies the bound in Proposition \ref{prop:SingletonRM} with equality. We also say that the vector code $C$ is MRD if $\Gamma(C)$ is MRD for any basis $\Gamma$ of $\Fqm$ over $\Fq$. It easy to see that $C$ is MRD if and only if $d(C)=n-\dimqm(C)+1$.
Delsarte, in \cite[Theorem~5.5]{delsarte1978bilinear}, proved that $\C$ is an MRD code if and only if $\C^\perp$ is MRD and that such codes exist for every choice of the parameters $n,m$ and $d$. 

Delsarte \cite{delsarte1978bilinear}, Gabidulin \cite{gabidulin1985theory}, and Roth \cite{roth1991maximum} found independent constructions of MRD vector codes for any choice of the parameters. These constructions were then generalized in \cite{lunardon2018generalized} and \cite{sheekey2016new}.  

Recall that for any $\theta \in \Fqm$, the {\bf norm} of $\theta$ over $\Fq$ is defined by 
$N_{q^m/q}(\theta):=\theta^{\frac{q^m-1}{q-1}}$. 
We now set up some notation to handle code equivalence, based on the approach in \cite{schmidt2018number}.

Let $f\in\L$ be a linearized polynomial and let $\beta=(\beta_1,\ldots,\beta_n) \in \Fqm^n$ have coefficients that are linearly independent over $\Fq$. 
We define $\pi_\beta(f)$ as follows: 
\begin{align*}
    \pi_\beta:\L \longrightarrow \F_{q}^{n\times m} :
    f\longmapsto \Gamma\left(f(\beta_1),\ldots,f(\beta_n)\right).
\end{align*}

Let $U$ be a subspace of $\Fqm$ over $\Fq$ and let $\gamma = (\gamma_1,\ldots,\gamma_n) \in \Fqm^n$ such that $\{\gamma_1,\ldots,\gamma_n\}$ is a basis of $U$. Let $I_U$ be the ideal in $\L$ of linearized polynomials that vanish on $U$ and let $\pi_U:\L \longrightarrow \L/I_U$ be the natural quotient map. Let $f,g \in \L$. Then $\pi_U(f)=\pi_U(g)$ if and only if $\pi_\gamma(f) = \pi_\gamma(g)$. In particular, for any subspace $V$ of $\L$,  $\pi_U(V)$ and $\pi_{\gamma}(V)$ are isomorphic \cite[Lemma 2.1]{schmidt2018number}.

\begin{definition}
\label{def:Gab}
	Let $k,s$ be integers such that $1\leq k < n$ and $\gcd(s,m)=1$. Let $\eta\in\Fqm$ such that $N_{q^{ms}/q^s}(\eta)\neq (-1)^{mk}$ and define the set
	\begin{equation*}
	\G_{k,s}(\eta):=\left\{\sum_{i=0}^{k-1} f_i\,x^{q^{si}}+\eta\, f_0\,x^{q^{sk}}:f_0,\ldots,f_{k-1}\in\Fqm\right\}.
	\end{equation*}
    A $k$-dimensional ($\Fqm$-\textbf{linear}) \textbf{generalized twisted Delsarte-Gabidulin code} is defined to be a code of the form $\pi_U\left(\G_{k,s}(\eta)\right)$ for some $n$-dimensional subspace $U$ of $\Fqm$.
    A code $\pi_U\left(\G_{k,s}(0)\right)$ is called a \textbf{Delsarte-Gabidulin} code.
\end{definition}

We list some facts on Delsarte-Gabidulin codes. Let $U$ be an $n$-dimensional subspace of $\Fqm$ and let $\alpha$ be a primitive element of $\Fqm$ over $\Fq$. In the notation of Definition \ref{def:Gab}, the following hold.

\begin{enumerate}[leftmargin=5.5ex]
     \item $\pi_U\left(\G_{1,s}(\eta)\right)$ is a $1$-dimensional Delsarte-Gabidulin code.
     \item The dual of $\pi_U\left(\G_{k,s}(\eta)\right)$ is equivalent to $\pi_U\left(\G_{n-k,s}(-\eta)\right)$. See \cite[Proposition~4.2]{lunardon2018generalized}.
    \item $\pi_U\left(\G_{k,s}(\eta)\right)$ is MRD. See \cite[Theorem~3.3]{lunardon2018generalized} and \cite[Theorem~5]{sheekey2016new} for a proof.
     \item $\pi_U\left(\G_{k,s}(0)\right)$ and $\pi_V\left(\G_{k,s}(0)\right)$ are equivalent
     if and only if $V= \{ \lambda u^{q^\ell} : u \in U \}$ for some $\lambda \in \Fqm, \ell \in \Z$. See \cite[Theorem~3.1]{schmidt2018number}.  
    \item $\pi_{\<1,\alpha,\ldots,\alpha^{n-1}\>}\left(\G_{k,s}(0)\right)$ is equivalent to $\pi_{\<1,\alpha,\ldots,\alpha^{n-1}\>}\left(\G_{k,1}(0)\right)$. See \cite[Theorem~3.1]{schmidt2018number}.
    \item If $\dimq(U)\in\{m-1,m\}$ then $\pi_U\left(\G_{k,s}(0)\right)$ is equivalent to $\pi_{\<1,\alpha,\ldots,\alpha^{n-1}\>}\left(\G_{k,1}(0)\right)$. See \cite[Corollary~3.2]{schmidt2018number}.
\end{enumerate}

While many properties of $\pi_U\left(\G_{k,s}(\eta)\right)$ have been studied, we know very little about the tensor rank of such codes; it is not yet fully known for which $k$ and $U$ that $\pi_U\left(\G_{k,s}(\eta)\right)$ is MTR. 
\newline

In the remainder we let $\alpha$ be primitive element of $\Fqm$ over $\Fq$ and $U$ be the $n$-dimensional subspace of $\Fqm$ generated by $\{1,\alpha,\ldots,\alpha^{n-1}\}$.
\newline

We now give some results that establish the tensor rank of the code $\pi_U\left(\G_{1,1}(0)\right)$ and its dual for all values of $m$ and for any primitive element $\alpha$ of $\Fqm$ over $\Fq$. These results can be extended to any equivalent code, as the tensor rank is an invariant of tensor equivalence. 

\begin{proposition}[{\cite[Corollary~5.14]{byrne2019tensor}}]
\label{prop:MTR1dim}
    Let $q\geq m+n-2$. We have that  $\pi_U\left(\G_{1,1}(0)\right)$ has tensor rank $m+n-1$ and, in particular, it is MTR.
\end{proposition}

As a consequence of Lemma \ref{lem:lindep} and Proposition \ref{prop:MTR1dim}, we have the following results.

\begin{corollary}
\label{cor:dim2}
    Let $q\geq 2m-3$ and $V\leq\Fqm^m$ be the row-space of 
    \begin{equation*}
         G:=\begin{pmatrix}
            \beta_{1,1} & \cdots & \beta_{1,m}\\
            \beta_{2,1} & \cdots & \beta_{2,m}
        \end{pmatrix},
    \end{equation*}
    for some linearly independent sets $\{\beta_{1,1},\ldots,\beta_{1,m}\}$ and $\{\beta_{2,1},\ldots,\beta_{2,m}\}$ of elements in $\Fqm$. We have $  3m-2\leq\trk(V) \leq 4m-4$.
\end{corollary}
\begin{proof}
    It is immediate to observe that $V$ is also the row-space of the reduced row echelon form $\overline{G}$ of $G$, i.e.
    \begin{equation*}
        \overline{G}:=
        \begin{pmatrix}
            1 & 0 & \gamma_{1,1} & \cdots & \gamma_{1,m-2}\\
            0 & 1 & \gamma_{2,1} & \cdots & \gamma_{2,m-2}
        \end{pmatrix},
    \end{equation*}
    for some linearly independent sets $\{\gamma_{1,1},\ldots,\gamma_{1,m-2}\}$ and $\{\gamma_{2,1},\ldots,\gamma_{2,m-2}\}$ of elements in $\Fqm$. 
    Observe that $\pi_{\<1, \gamma_{1,1},\ldots, \gamma_{1,m-2}\>}\left(\G_{1,1}(0)\right)$ and $\pi_{\<1, \gamma_{2,1},\ldots,\gamma_{2,m-2}\>}\left(\G_{1,1}(0)\right)$ are both equivalent to $\pi_{U}\left(\G_{1,1}(0)\right)$.
    Lemma \ref{lem:lindep} and Proposition \ref{prop:MTR1dim} imply
    \begin{equation*}
        \trk(V)=\trk\left(\<\overline{G}_1\>+\<\overline{G}_2\>\right)\leq \trk\left(\<\overline{G}_1\>\right)+\trk\left(\<\overline{G}_1\>\right)=(2m-2)+(2m-2)=4m-4.
    \end{equation*}
    Finally, one can check that $d(V)=m-1$ and therefore, by the tensor rank bound, we get 
    \begin{equation*}
        3m-2= \dimq(V)+d(V)-1\leq \trk(V),
    \end{equation*}
    which concludes the proof.
\end{proof}

\begin{corollary}
    Let $1\leq k\leq n\leq m$, $1\leq i\leq n-k$ and $j_1,\ldots,j_k$ be integers. Let $q\geq m+i-1$. For each $s \in [k]$, let $\alpha_s$ be a primitive element of $\Fqm$ over $\Fq$. Let $V\leq\Fqm^n$ the row-space of
    \begin{equation*}
        G:=\left(\begin{array}{*{7}c}
            1 & 0 & \cdots & 0 & \beta_{1,1} & \cdots & \beta_{1,n-k}\\
            0 & 1 & \cdots & 0 & \beta_{2,1} & \cdots & \beta_{2,n-k}\\
            & & \ddots & & \vdots & & \vdots\\
            0 & 0 & \cdots & 1 & \beta_{k,1} & \cdots & \beta_{k,n-k}
        \end{array}\right),
    \end{equation*}
    where $\<1,\beta_{s,1},\ldots,\beta_{s,n-k}\> = \<1,\alpha_s^{q^{j_s}},\alpha_s^{2q^{j_s}},\ldots,\alpha_s^{iq^{j_s}}\>$ for all $s \in [k]$ and $r\in [n-k]$. We have $\trk(V)\leq k(m+i)$.
\end{corollary}
\begin{proof}
    The proof follows as a straightforward consequence of Lemma \ref{lem:lindep} and Proposition \ref{prop:MTR1dim}. In particular, for all $s\in k$, Lemma \ref{lem:lindep} implies that 
    \begin{equation*}
        \trk\left(\<\left(0,\ldots,0,1,0,\ldots,0, \beta_{s,1},\ldots,\beta_{s,n-k}\right)\>\right)=\trk\left(\<\left(1, \alpha_s^{q^{j_s}},\ldots,\alpha_s^{iq^{j_s}}\right)\>\right),
    \end{equation*}
   and, by Proposition \ref{prop:MTR1dim}, we have
   \begin{equation*}
       \trk\left(\<\left(1, \alpha_s^{q^{j_s}},\ldots,\alpha_s^{iq^{j_s}}\right)\>\right)=m+i.
   \end{equation*}
   Therefore, we get
   \begin{equation*}
       \trk(V)=\trk\left(\<\sum_{s=1}^k G_s\>\right)\leq \sum_{s=1}^k\trk\left(\<G_s\>\right)=k(m+i),
   \end{equation*}
   where $G_s$ is the $s$th row of $G$, $s\in[k]$. This concludes the proof.
\end{proof}

\begin{remark}
    In \cite[Proposition~5.15]{byrne2019tensor} it was shown that 
    \begin{equation*}
        \trk\left( \pi_U\left(\G_{k,1}(0)\right) \right) \leq k(m+n-1).
    \end{equation*}
    The bound in Corollary \ref{cor:dim2} can be applied to give the improvement on this estimate for $k=2$ and $n=m$.
\end{remark}

\begin{theorem}
\label{thm:MTR(n-1)dim}
    Let $q\geq m$. We have that $\pi_U\left(\G_{1,1}(0)\right)^\perp$ has tensor rank $mn-m+1$ and, in particular, it is MTR.
\end{theorem}
\begin{proof}
    Let $\Gamma:=\{1,\alpha,\ldots,\alpha^{m-1}\}$ be an $\Fq$-basis of $\Fqm$ where $\alpha$ is a primitive element of $\Fqm$ over $\Fq$. For ease of notation, write $\C:=\Gamma\left(\pi_U\left(\G_{1,1}(0)\right)\right)^\perp$ for the remainder of the proof. Let $\mS\subseteq\Fq^\times$ be a set of $m-1$ distinct elements and let $M$ be the companion matrix of the characteristic polynomial of $\alpha$. Corollary \ref{cor:basenm} implies that the set 
    \begin{equation*}
		\overline{\A}:=\{Y_nJ^i\,E_{1,j}\,\left(M^{-i}\right)^t:s+1\leq j\leq m,0\leq i\leq n-1\}\,\cup\,\{\,Y_nJ^i\,\E(\gamma)\,\left(M^{-i}\right)^t:0\leq i\leq n-2,\gamma\in\mS\}.
	\end{equation*}
	is an $(nm-m+1)$-base for $\<Y_n,Y_nM,\ldots,Y_nM^{m-2}\>^\perp$. Observe that, 
	\begin{equation*}
	    \<Y_n,Y_nM,\ldots,Y_nM^{m-2}\>\leq \<Y_n,Y_nM,\ldots,Y_nM^{m-1}\>=\C^\perp
	\end{equation*}
	which implies that $\C\leq\<Y_n,Y_nM,\ldots,Y_nM^{m-2}\>^\perp$ and therefore $\A$ is also a perfect base for $\C^\perp$. By the tensor rank bound, we have
	\begin{equation*}
	    m(n-1)+2-1=\dimq\left(\C\right)+d\left(\C\right)-1\leq \trk\left(\C\right)\leq nm-m+1,
	\end{equation*}
	which implies that $\trk(\C)=mn-m+1$ and therefore $\C$ is MTR. The statement follows from the fact that the tensor rank is invariant under equivalence.
\end{proof}

Denote by $\maxtrk(n,m,k)$ the maximum tensor-rank for $X\in\Kknm$. It is well-known that this function is symmetric in $m,n,k$. See, for example, \cite{atkinson1979maximal,dobkinphd,brockett1978optimal} for further details on this function.

\begin{proposition}[{{\cite[Theorem~2]{atkinson1979maximal}}}]
\label{prop:maxtrk}
	Let $k\leq n$. We have $$\maxtrk(n,m,mn-k)=mn-k^2+\maxtrk(k,k,k^2-k).$$
\end{proposition}

\begin{remark}
\label{rem:atkmax}
	In \cite{atkinson1979maximal}, the authors stated that $\maxtrk(3,3,6)=7$ and $\maxtrk(4,4,12)=14$. Moreover they showed that $\maxtrk(2,2,2)=3$, $\maxtrk(3,3,5)\in\{6,7\}$,
	\begin{align*}
		\maxtrk(3,3,i)&=
		\begin{cases}
			i+2 \textup{ if }\;i\in\{1,2,3,4\},\\
			i+1 \textup{ if }\;i\in\{6,7\},\\
			i  	\textup{ if }\;i\in\{8,9\},\\
		\end{cases}\\
		\maxtrk(n,m,nm-i)&=
		\begin{cases}
			mn-1, \textup{ if }\;i\in\{1,2\}\;\textup{and}\;m,n\geq i,\\
			mn-2, \textup{ if }\;i\in\{3,4\}\;\textup{and}\;m,n\geq i.
		\end{cases}
	\end{align*}
\end{remark}

In Table \ref{table:G(m,a)}, we summarize tensor ranks for the Delsarte-Gabidulin codes that are equivalent to $\pi_U\left(\G_{k,1}(0)\right)$.
\newline

\renewcommand\arraystretch{1.3}
\setlength\LTleft{0pt}
\setlength\LTright{0pt}
\begin{longtable}[c]{@{\extracolsep{\fill}}*6{|c}|@{}}
		\hline
		$n$ & $m$ & $k$ & $q\geq$ & $\trk$ & \temp{Reference} \\\hline\hline
		\multirow{5}{*}{$2$}& \multirow{3}{*}{$2$} &  \multirow{2}{*}{$1$} & \multirow{2}{*}{$2$} & \multirow{2}{*}{$3$} & \temp{\cite{atkinson1979maximal} $-$ \cite{byrne2019tensor} $-$ \cite{jaja1979optimal}} \\
		& & & & & \temp{Corollary \ref{cor:n=2,3} and Remark \ref{rem:m}}\\\cline{3-6}
		& & $2$ & $2$ & $4$ & \temp{trivial}\\\cline{2-6}
		& \multirow{2}{*}{$\geq 3$} & $1$ & $m$ & $m+1$ & \temp{\cite{byrne2019tensor} $-$ Corollary \ref{cor:n=2,3} and Remark \ref{rem:m}}\\\cline{3-6}
		& & $2$ & $2$ & $2m$ & \temp{trivial}\\\hline
		\multirow{9}{*}{$3$} & \multirow{3}{*}{$3$} & \multirow{2}{*}{$1$} & \multirow{2}{*}{$4$} & \multirow{2}{*}{$5$} & \temp{\cite{atkinson1979maximal} $-$ \cite{byrne2019tensor}}\\
		& & & & & \temp{Corollary \ref{cor:n=2,3} $-$ Proposition \ref{prop:MTR1dim}}\\\cline{3-6}
		& & $2$ & $3$ & $7$ & \temp{\cite{atkinson1979maximal} $-$ Theorem \ref{thm:MTR(n-1)dim}}\\\cline{3-6}
		& & $3$ & $2$ & $9$ & \temp{trivial}\\\cline{2-6}
		& \multirow{3}{*}{$4$} & $1$ & $5$ & $6$ & \temp{\cite{byrne2019tensor} $-$ Corollary \ref{cor:n=2,3} $-$ Proposition \ref{prop:MTR1dim}}\\\cline{3-6}
		& & $2$ & $4$ & $9$ & \temp{Theorem \ref{thm:MTR(n-1)dim}}\\\cline{3-6}
		& & $3$ & $2$ & $12$ & \temp{trivial}\\\cline{2-6}
		& \multirow{3}{*}{$\geq 5$} & $1$ & $m+1$ & $m+2$ & \temp{\cite{byrne2019tensor} $-$ Corollary\ref{cor:n=2,3} $-$ Proposition \ref{prop:MTR1dim}}\\\cline{3-6}
		& & $2$ & $m$ & $2m+1$ & \temp{Theorem \ref{thm:MTR(n-1)dim}}\\\cline{3-6}
		& & $3$ & $2$ & $3m$ & \temp{trivial}\\\hline
		\multirow{8}{*}{$4$} & \multirow{4}{*}{$4$} & $1$ & $6$ & $7$ & \temp{\cite{byrne2019tensor}$-$ Proposition \ref{prop:MTR1dim}}\\\cline{3-6}
		& & $2$ & $5$ & $\leq 12$ & \temp{Corollary \ref{cor:dim2}}\\\cline{3-6}
		& & $3$ & $4$ & $13$ & \temp{Theorem \ref{thm:MTR(n-1)dim}}\\\cline{3-6}
		& & $4$ & $2$ & $16$ & \temp{trivial}\\\cline{2-6}
		& \multirow{4}{*}{$\geq 5$} & $1$ & $m+2$ & $m+3$ & \temp{\cite{byrne2019tensor}$-$ Proposition \ref{prop:MTR1dim}}\\\cline{3-6}
		& & $2$ & $m+2$ & $\leq 2m+6$ & \temp{\cite{byrne2019tensor}}\\\cline{3-6}
		& & $3$ & $m$ & $3m+1$ & \temp{Theorem \ref{thm:MTR(n-1)dim}}\\\cline{3-6}
		& & $4$ & $2$ & $4m$ & \temp{trivial}\\\hline
		\multirow{13}{*}{$\geq 5$} & \multirow{7}{*}{$n$} & $1$ & $2m-2$ & $2m-1$ & \temp{\cite{byrne2019tensor}$-$ Proposition \ref{prop:MTR1dim}}\\\cline{3-6}
		& & $2$ & $2m-3$ & $\leq 4m-4$ & \temp{Corollary \ref{cor:dim2}}\\\cline{3-6}
		& & $3$ & $2m-2$ & $\leq 6m-6$ & \temp{\cite{byrne2019tensor}}\\\cline{3-6}
		& & $\vdots$ & $\vdots$ & $\vdots$ & \temp{$\vdots$}\\\cline{3-6}
		& & $n-2$ & $2m-2$ & $\leq (m-2)(2m-1)$ & \temp{\cite{byrne2019tensor}}\\\cline{3-6}
		& & $n-1$ & $m$ & $2m^2-m+1$ & \temp{Theorem \ref{thm:MTR(n-1)dim}}\\\cline{3-6}
		& & $n$ & $2$ & $m^2$ & \temp{trivial}\\\cline{2-6}
		& \multirow{6}{*}{$>n$} & $1$ & $m+n-2$ & $m+n-1$ & \temp{\cite{byrne2019tensor}}\\\cline{3-6}
		& & $2$ & $m+n-2$ & $\leq 2(m+m-1)$ & \temp{\cite{byrne2019tensor}}\\\cline{3-6}
		& & $\vdots$ & $\vdots$ & $\vdots$ & \temp{$\vdots$}\\\cline{3-6}
		& & $n-2$ & $m+n-2$ & $\leq (n-2)(m+n-1)$ & \temp{\cite{byrne2019tensor}}\\\cline{3-6}
		& & $n-1$ & $m$ & $nm-m+1$ & \temp{Theorem \ref{thm:MTR(n-1)dim}}\\\cline{3-6}
		& & $n$ & $2$ & $nm$ & \temp{trivial}\\\hline
	\caption{\label{table:G(m,a)} Tensor rank of codes equivalent to $\pi_U\left(\G_{k,1}(0)\right)$.}
\end{longtable}

Our results establish the existence of MTR codes for some parameter sets. In particular, for all $n,m$, we have the following observations.

\begin{itemize}[leftmargin=5.5ex]
	\item Let $q\geq m+n-2$. There exists an MTR code of dimension $m$, namely $\pi_U\left(\G_{1,1}(0)\right)$.
	\item Let $q\geq m$. There exists an MTR code of dimension $(n-1)m$, namely $\pi_U\left(\G_{1,1}(0)\right)^\perp$. 
\end{itemize}

We conclude this section by deriving a result on the existence of some families of MTR codes. We will prove this result by exploiting the connection between rank-metric codes and block codes in the Hamming metric introduced by \cite{brockett1978optimal} (see also \cite[Chapter~18]{burgisser2013algebraic}). For an $R$-base $\A:=\{A_1,\ldots,A_R\}\subseteq\Fqnm$, we define, as in \cite[Section~4.1]{byrne2019tensor}, the $\Fq$-linear isomorphism
\begin{equation*}
	\psi_\A:\<\A\>\longrightarrow\Fq^R\,:\,\sum_{i=1}^R\lambda_iA_i\longmapsto\sum_{i=1}^R\lambda_ie_i,
\end{equation*}
	where $\{e_1,\ldots,e_R\}$ is the standard basis of $\Fq^R$. 
	
\begin{definition}[{{\cite[Definition~4.10]{byrne2019tensor}}}]
	Let $\C$ be a code with tensor-rank $R$ and let $\A$ be an $R$-base for $\C$. We define the linear block code $C_\A$ to be the image of $\C$ under $\psi_\A$, i.e. $C_\A:=\psi_\A(\C)$, endowed with the Hamming distance.
\end{definition}

We recall that a nonzero linear block code $D\leq\Fq^n$ is said to be \textbf{MDS} (\textbf{maximum distance separable}) if it attains the Singleton bound, i.e. if $d^\H(D)=n-\dimq(D)+1$, where $d^\H(D):=\min(\{w^\H(x):x\in D, x \neq 0\})$ and, for all $x\in D$, $w^\H(x):=|\{i: x_i \neq 0\}|$.

\begin{remark}
	We recall some well-known facts on linear block codes, see \cite{pless1998introduction} for further details. The shortened code $D$ of an $\Fq$-$[R,k,d]$ linear block code $C\in\Fq^R$ on the set $S\subseteq\{1,\ldots,R\}$ is defined as $	D:=\{(c_i)_{i\in S}:c \in C, \supp(c) \subseteq S\}$. $D$ is an $\Fq$-$[|S|,k-|S^c|,d]$ if $|S^c|< d-1$.
\end{remark}

\begin{lemma}
\label{lem:subMTR}
	Let $\C$ be an $\Fq$-$[n\times m,k,n]$ code with $k\leq n$ and tensor rank $R=k+n-1$. Let $\A:=\{A_1,\ldots,A_{R}\}$ be a $R$-base for $\C$. Let $S\subseteq [R]$. Then either
	\begin{enumerate}[label={(\arabic*)},leftmargin=5.5ex]
		\item \label{item1:lemMTR} $\C\cap\<A_s\,:\,s\in S\>=\{0\}$ and $0\leq|S|\leq n-1$, or
		\item \label{item2:lemMTR} $\C\cap\<A_s\,:\,s\in S\>$ is an $\Fq$-$[n\times m,|S|-n+1,n]$ MTR code and $n\leq |S|\leq R$.
	\end{enumerate}  
\end{lemma}
\begin{proof}
	If $|S|\in\{0,R\}$ the result is trivial. Let $1\leq|S|\leq n-1$ and suppose towards a contradiction that there exists a nonzero matrix $N\in\C\cap\<A_s\,:\,s\in S\>$. There exist $\lambda_1,\ldots,\lambda_{|S|}\in\Fq$ such that 
	\begin{equation*}
		n\leq \rk(N)=\rk\left(\sum_{i=1}^{|S|}\lambda_iA_i\right)\leq \sum_{i=1}^{|S|}\lambda_i\;\rk(A_i)\leq n-1,
	\end{equation*}
	which leads to a contradiction. This proves \ref{item1:lemMTR}. We now assume that $n\leq |S|\leq R-1$ in the remainder. 
	Let $D\leq \Fq^{R}$ be the code obtained by shortening $C_\A$ on $S$. Observe that $C_\A$ is an $\Fq$-$[R,k,n]$ MDS linear block code, and therefore $D$ is an $\Fq$-$[|S|,|S|-n+1,n]$ MDS linear block code, since $R-|S|\leq k-1<n$. It follows that 
	\begin{equation*}
		\D:=\psi_{\{A_s\,:\,s\in S\}}^{-1}(D)
	\end{equation*}
	is an $\Fq$-$[n\times m,|S|-n+1,n]$ MTR code. Then $\D=\C\cap\<A_s\,:\,s\in S\>$, since $\D$ contains all the elements of $\C$ that can be written as linear combination of the elements of $\{A_s\,:\,s\in S\}$. This implies \ref{item2:lemMTR} and concludes the proof.
\end{proof}

In the following result we use Lemma \ref{lem:subMTR} to show  that an MTR code exists for some values of $m,n,k,q$.

\begin{theorem}
	Let $m,n,k,d\in\Z_{>0}$ and $q$ power prime be such that $1\leq k\leq m$, $d\leq n$ and $q\geq m+d-2$. There exists an $\Fq$-$[n\times m,k,d]$ MTR code.  
\end{theorem}
\begin{proof}
	Let $\Gamma:=\{1,\alpha,\ldots,\alpha^{m-1}\}$ be an $\Fq$-basis of $\Fqm$ where $\alpha$ is a primitive element of $\Fqm$ over $\Fq$. Let $\C:=\Gamma\left(\pi_{\<1,\alpha,\ldots,\alpha^{d-1}\>}\left(\G_{1,1}(0)\right)\right)$ is an $\Fq$-$[d\times m,m,d]$ MTR code by Proposition \ref{prop:MTR1dim}. By Lemma \ref{lem:subMTR} there exists an $\Fq$-$[d\times m,k,d]$ MTR subcode $\D$ of $\C$. Using the idea as in the proof of Lemma \ref{lem:lindep}, we can extend $\D$ to an $\Fq$-$[n\times m,k,d]$ code $\overline\D$ by adding $n-d$ rows to the matrices of $\D$ without increasing their ranks. The modified code $\overline\D$ retains the MTR property, which concludes the proof.
\end{proof}

\section{Acknowledgements}
  The authors are very grateful to John Sheekey for useful discussions on this topic.

\bibliographystyle{amsplain} 
\bibliography{trbiblio.bib}
\end{document}